\documentclass[a4paper,11pt]{article}

\usepackage{mathtools}  
\usepackage{mathrsfs}
\usepackage{amsmath}
\usepackage{amssymb}
\usepackage[margin=0.75in]{geometry}
\usepackage{textcomp}
\usepackage{array}
\usepackage[usenames, dvipsnames]{color}
\usepackage[font=small]{caption}
\usepackage{floatrow}
\usepackage{cite}
\usepackage[utf8]{inputenc}
\usepackage{amsthm}

\numberwithin{equation}{section}

\newtheorem{lemma}{Lemma}[section]

\usepackage[hyperfootnotes=false, linktocpage=true, colorlinks, citecolor=blue, linkcolor=blue, urlcolor=Maroon]{hyperref}

\newcommand{\mc}{\mathcal}
\newcommand{\mbb}{\mathbb}

\newcommand{\D}{D}

\newcommand{\smlhdg}[1]{\bigskip\noindent\textbf{#1}\medskip}
\newcommand{\smlhdgnogap}[1]{\noindent\textbf{#1}\medskip}

\newcolumntype{C}{>{$}c<{$}} 		

\newcommand{\be}{\begin{equation}}
\newcommand{\ee}{\end{equation}}
\newcommand{\bea}{\begin{eqnarray}}
\newcommand{\eea}{\end{eqnarray}}
\newcommand{\ba}{\begin{align}}
\newcommand{\ea}{\end{align}}
\newcommand{\bi}{\begin{itemize}}
\newcommand{\ei}{\end{itemize}}

\newtheorem{thm}{Theorem}[section]
\newtheorem{crl}{Corollary}[section]

\newcommand{\ceil}[1]{\mathrm{ceil}\left(#1\right)}	
\newcommand{\iso}{\cong}							

\newcommand{\moricn}{\mathcal{M}}	
\newcommand{\nefcn}{\mathcal{N}}		
\newcommand{\hcl}{l} 					
\newcommand{\ecl}{e}					
\newcommand{\lb}{L} 					
\newcommand{\dual}[1]{#1^*}			
\newcommand{\ind}{\mathrm{ind}}		

\newcommand{\surf}{S}				
\newcommand{\surfdegr}{d}			

\newcommand{\tc}{z}					

\newcommand{\shfdiv}[1]{\tilde{#1}} 					
\newcommand{\shfdivstab}[1]{\tilde{\underline{#1}}} 	
\newcommand{\dps}[1]{\mathrm{dP}_{#1}}			

\begin{document}

\begin{centering}
\vspace*{1.2cm}
{\Large \bf Index Formulae for Line Bundle Cohomology on Complex Surfaces}

\vspace{1cm}

{\large
{\bf{Callum R. Brodie}$^{a,}$\footnote{callum.brodie@physics.ox.ac.uk}},  
{\bf{Andrei Constantin}$^{b,}$\footnote{andrei.constantin@mansfield.ox.ac.uk}},
{\bf{Rehan Deen}$^{a,}$\footnote{rehan.deen@physics.ox.ac.uk }},    
{\bf{Andre Lukas}$^{a,}$\footnote{lukas@physics.ox.ac.uk}},
\bigskip}\\[0pt]
\vspace{0.23cm}
${}^a$ {\it 
Rudolf Peierls Centre for Theoretical Physics, University of Oxford,\\
Parks Road, Oxford OX1 3PU, UK
}
\\[2ex]
${}^b$ {\it 
Pembroke College, University of Oxford, OX1 1DW, UK \\
$~~$Mansfield College, University of Oxford, OX1 3TF, UK
}

\begin{abstract}\noindent
We conjecture and prove closed-form index expressions for the cohomology dimensions of line bundles on del Pezzo and Hirzebruch surfaces. Further, for all compact toric surfaces we provide a simple algorithm which allows expression of any line bundle cohomology in terms of an index. These formulae follow from general theorems we prove for a wider class of surfaces. In particular, we construct a map that takes any effective line bundle to a nef line bundle while preserving the zeroth cohomology dimension. For complex surfaces, these results explain the appearance of piecewise polynomial equations for cohomology and they are a first step towards understanding similar formulae recently obtained for Calabi-Yau three-folds. 
\end{abstract}
\end{centering}

\newpage
\tableofcontents
\newpage


\section{Introduction}
\label{sec:intro_and_sum}

Cohomologies of line bundles are crucial for various types of string compactifications, for example in the context of heterotic  and F-theory model building. Usually, these cohomologies are computed using algorithmic methods, based on \v{C}ech cohomology, spectral sequences and related mathematical tools. These methods can be computationally intense and they provide little insight into the origin and structure of the results. This makes a bottom-up approach to string model building difficult whenever line bundles are involved. Clearly, string theory would profit from more direct and systematic access to line bundle cohomology and in this paper we report on some progress in this direction.

In Ref.~\cite{Constantin:2018hvl} it was found that line bundle cohomology dimensions on (complete intersection) Calabi-Yau manifolds can be described by relatively simple formulae, which are piecewise polynomial.\footnote{The first non-trivial instance of a line bundle cohomology formula appeared in the earlier work \cite{Constantin:2018otr, Buchbinder:2013dna}, in which generic hypersurfaces of type (2,2,2,2) in a product of four complex projective spaces were studied.}  More precisely, the Picard group of the manifold splits into a number of disjoint regions - which are frequently but not always cones - in each of which the cohomology dimensions are given by a cubic polynomial in the integers which label the line bundles. These results were obtained heuristically by looking at algorithmically computed cohomology data on a few complete intersection Calabi-Yau manifolds and smooth quotients thereof and by extracting analytic formulae from this data. The authors of Ref.~\cite{Klaewer:2018sfl} employed machine learning techniques to derive cohomology formulae for (hypersurfaces in) toric varieties. In a recent paper~\cite{Larfors:2019sie}, these results were extended to a larger class of complete intersection manifolds with Picard number $2$. The pattern conjectured from these results is that line bundle cohomology dimensions on manifolds with complex dimension $n$ are piecewise polynomial, with polynomials of degree (at most) $n$. All this points to a yet to be fully discovered mathematical structure underlying line bundle cohomology.

 In simple cases, it is possible to prove these results by spectral sequence chasing but this becomes tedious for more complicated manifolds. It is fair to say that the origin of these formulae for Calabi-Yau three-folds is currently not well-understood and that there is no simple and systematic method for their derivation. In fact, Calabi-Yau three-folds are fairly complicated objects and may not provide the best setting to uncover the structure underlying line bundle cohomology.

In the present paper, we, therefore, explore these issues in the simpler setting of complex surfaces. A line bundle $L\rightarrow S$ over a complex surface $\surf$ has three cohomology dimensions $h^q(\surf,\lb)$, where $q=0,1,2$, but knowledge of one of these for all line bundles implies the other two via Serre duality and the index theorem. For this reason, we will focus our study on the zeroth cohomology dimension, $h^0(\surf,\lb)$. Following standard notation, the line bundle associated to a divisor $D\subset S$ is denoted by $\mc{O}_\surf(D)$ and frequently we introduce a basis $D_i$ of divisor (classes), so that $D=\sum_ik_i D_i$, with $k_i\in\mathbb{Z}$. Having fixed such a basis, we also occasionally write the corresponding line bundle as $\mc{O}_\surf({\bf k}):=\mc{O}_\surf(D)$, where ${\bf k}=(k_i)$ is an integer vector.

As a first step, we proceed in much the same way as was done for Calabi-Yau three-folds. We produce cohomology data from algorithmic methods and extract analytic formulae for the zeroth cohomology dimension by ``eyeballing". The examples we will be studying in this context are the Hirzebruch and del Pezzo surfaces. The results turn out to be as expected and are analogous to the ones obtained for Calabi-Yau three-folds. The dimensions $h^0(\surf,\mc{O}_\surf({\bf k}))$ are described by equations which are piecewise quadratic in the integers $k_i$.  This provides further evidence that piecewise polynomial formulae for line bundle cohomology dimensions are a general feature which is neither limited to the Calabi-Yau case nor to three complex dimensions. 

However, for surfaces we are able to go further. For both Hirzebruch and del Pezzo surfaces we are able to express the piecewise quadratic formulae in terms of basis-independent, intrinsic objects, specifically the generators of the effective and nef cones, the irreducible negative self-intersection divisors and the intersection form on $\surf$. This re-writing suggests a specific map $D\rightarrow\tilde{D}$ for effective divisors which we provide explicitly and which leaves the zeroth cohomology dimension unchanged, that is, $h^0(\surf,\mc{O}_\surf(D))=h^0(\surf,\mc{O}_\surf(\tilde{D}))$. We further observe that for all Hirzebruch and del Pezzo surfaces the line bundle $\mc{O}_\surf(\tilde{D})$ satisfies $h^q(\surf,\mc{O}_\surf(\tilde{D}))=0$ for $q=1,2$. Combining these statements implies that the cohomology of $\mc{O}_\surf(D)$ can be computed in terms of the index of the shifted divisor $\tilde{D}$ as
\begin{equation}
 h^0(\surf,\mc{O}_\surf(D))=h^0(\surf,\mc{O}_\surf(\tilde{D}))={\rm ind}(\mc{O}_\surf(\tilde{D}))\, . \label{emp}
\end{equation}
At this stage, the result~\eqref{emp} for Hirzebruch and del Pezzo surfaces together with the map $D\rightarrow \tilde{D}$ is still empirical, that is, inferred from a finite set of cohomology data.

Based on these empirical results, we write down a general form of the map $D\rightarrow \tilde{D}$ which applies to all smooth compact complex projective surfaces and we prove that it preserves the zeroth cohomology dimensions, that is, $h^0(\surf,\mc{O}_\surf(\tilde{D}))=h^0(\surf,\mc{O}_\surf(D))$.  Moreover, it follows that iterating the map $D\rightarrow \tilde{D}$ leads, after a finite number of steps, to a divisor $\underline{\tilde{D}}$ in the nef cone. Provided there is a vanishing theorem which asserts that $h^q(\surf,\mc{O}_\surf(\underline{\tilde D}))=0$ for $q=1,2$, it follows that 
\begin{equation}
 h^0(\surf,\mc{O}_\surf(D))=h^0(\surf,\mc{O}_\surf(\underline{\tilde D}))={\rm ind}(\mc{O}_\surf(\underline{\tilde D}))\, .  \label{indform}
\end{equation} 
For such cases, which we show include Hirzebruch and del Pezzo surfaces as well as all compact toric surfaces we have, therefore, a mathematical proof for the existence of index formulae for $h^0$ and a practical way of deriving them. These formulae are quasi-topological in nature: the quantity ${\rm ind}(\mc{O}_\surf(\underline{\tilde D}$)) is purely topological, while the map $D\rightarrow \tilde{D}$ depends in general on the complex structure. 

For Hirzebruch and del Pezzo surfaces we prove that a single application of the map $D\rightarrow \tilde{D}$ already projects into the nef cone and that a suitable vanishing theorem is available in either case. This leads to a mathematical proof for the empirical formula~\eqref{emp}.
\vspace{8pt}

We will presently provide a summary of our main results. Subsequently, the plan for the remainder of the paper is as follows. In the next section, we explain how cohomology formulae can be extracted from cohomology data, computed by algorithmic methods, focusing on Hirzebruch and del Pezzo surfaces. In a first instance, we extract piecewise quadratic formulae from the data which are subsequently refined to index formulae. The reader less interested in this ``empirical" aspect of the work can skip to Section~\ref{sec:genthms} which contains our main mathematical statements. Section~\ref{sec:app} illustrates these mathematical results in the context of simple examples. We conclude in Section~\ref{sec:con}.

The present paper has two companion papers. In Ref.~\cite{ml}, we explore how techniques from machine learning can help to uncover the structure of line bundle cohomology. Ref.~\cite{mathpaper} is more mathematical in style and provides rigorous proofs for the various mathematical statements presented here. 

\smlhdg{Summary of results}

\noindent 
We outline below the way in which the index formulae can be applied to specific surfaces.
For a smooth compact complex projective surface $\surf$ we first require knowledge of the effective (Mori) cone\footnote{For the examples discussed in this paper the Mori and the effective cones coincide, though there exist surfaces for which this is not the case, see e.g.~Example 1.5.1 in Ref.~\cite{lazarsfeld2004positivity}.}, ${\moricn}(S)$. In practice this amounts to providing the set of Mori cone generators $\hat{\moricn}(S)$ or the set of generators $\hat{\nefcn}(S)$ of the nef cone ${\nefcn}(S)$, which is dual to the Mori cone. We also need to know the intersection form $(D,D')\rightarrow D\cdot D'$ on $\surf$. For all divisors not in the Mori cone, that is $D\notin {\moricn}(S)$, we have $h^0(\surf,\mc{O}_\surf(D))=0$. On the other hand, all divisors $D\in{\moricn}(S)$ have strictly positive zeroth cohomology dimension. For such effective divisors we define the map $D\rightarrow \tilde{D}$ by
\be
\shfdiv{D} = D - \sum_{C\in{\cal I}} \theta( - D \cdot C ) \, \ceil{\frac{D \cdot C}{C^2}}C \,,  \label{Dmap}
\ee
where the sum runs over the set ${\cal I}$ of all irreducible curves with negative self-intersection. The Heaviside function $\theta$ ensures that only curves $C$ with $D\cdot C<0$ contribute to the sum and ${\rm ceil}$ is the ceiling function. Hence, to write down this map explicitly, we need to know the irreducible, negative self-intersection curves ${\cal}$ on the surface $\surf$ - information that can be obtained for many cases of interest. 

The key statement about the map~\eqref{Dmap} is that it leaves the zeroth cohomology dimension unchanged, that is, $h^0(\surf,\mc{O}_\surf(\tilde{D}))=h^0(\surf,\mc{O}_\surf(D))$. Since the nef cone ${\nefcn}(S)$ is the cone of divisors $D$ which intersect all algebraic curves non-negatively, it is clear that repeated application of the map~\eqref{Dmap} eventually leads to a divisor $\underline{\tilde{D}}$ in the nef cone. If there is a vanishing theorem, typically Kodaira vanishing or one of its refinements, which asserts that $h^q(\surf,\mc{O}_\surf(\underline{\tilde{D}}))=0$ for $q=1,2$ then the zeroth cohomology can be written as an index, using Eq.~\eqref{indform}. It turns out that this is the case for many surfaces of interest, including Hirzebruch surfaces, del Pezzo surfaces, and compact toric surfaces, and, hence, index formulae for the zeroth cohomology dimensions exist for all these cases. The relevant vanishing theorems will be reviewed in the main text.\\[2mm]
Let us first summarise how this general result applies to Hirzebruch surfaces $\mathbb{F}_n$.  The Picard lattice of all Hirzebruch surfaces is two-dimensional and we can introduce a basis $(D_1,D_2)$ of divisor classes, such that the intersection form is defined by $D_1^2=-n$, $D_1\cdot D_2=1$ and $D_2^2=0$ (see Appendix~\ref{app:hirz_surf} for details). Then, the Mori cone is generated by $D_1$ and $D_2$. This means effective divisors are of the form $D=k_1D_1+k_2D_2$ with non-negative $k_i$ and all other divisors with at least one $k_i$ negative have vanishing zeroth cohomology. The unique irreducible class with negative self-intersection is $C=D_1$. Inserting all this into the map~\eqref{Dmap} shows that $\tilde{D}$ is already in the nef cone. Theorem~\ref{demazure} guarantees the necessary vanishing, so that Eq.~\eqref{emp} can be applied. Combining these results we obtain the index formula
\be
h^0\left(\mbb{F}_n,\mc{O}_{\mbb{F}_n}(D)\right) = \ind\bigg( D - \theta(-D \cdot C) \, \ceil{\frac{D \cdot C}{C^2}}C \bigg) \,,
\ee
for the zeroth cohomology dimension of effective divisors $D$ on Hirzebruch surfaces. The explicit formula for the index of line bundles on Hirzebruch surfaces is provided in Eq.~\eqref{indFn}.\\[2mm]
Let us now summarise the analogous results for del Pezzo surfaces. Del Pezzo surfaces (other than $\mathbb{P}^1\times\mathbb{P}^1$ which is trivial in our context) are blow-ups  of the projective plane $\mathbb{P}^2$ in $n$ generic points, where $n=0,1,\ldots ,8$, and they are denoted by $\dps{n}$. The rank of their Picard lattice is $n+1$ and a standard basis of divisor classes consists of the hyperplane class $\hcl$ of $\mathbb{P}^2$ and the classes $\ecl_i$ of the exceptional divisors associated to the blow-ups, where $i=1,\ldots ,n$. The intersection form is fixed by the relation $l^2=1$, $l\cdot \ecl_i=0$ and $\ecl_i\cdot \ecl_j=-\delta_{ij}$. The list of generators of the Mori and nef cones is too long, at least for the larger values of $n$, to be listed here but has been explicitly provided in Appendix~\ref{app:dp_surf}. The irreducible negative self-intersection classes $C$ are precisely the Mori cone generators, which have self-intersection $C^2=-1$. It can be shown that a single application of the map~\eqref{Dmap} projects into the nef cone and Corollary~\ref{crldPvan} provides the appropriate vanishing statement so that Eq.~\eqref{emp} holds. This leads to the index formula
\be
h^0\left(\dps{n},\mc{O}_{\dps{n}}(D)\right) = \ind\bigg( D + \sum_{C\in\hat{\cal M}} \theta( - D \cdot C ) \, (D \cdot C)C \bigg) \,,
\ee
for the zeroth cohomology dimension of effective divisors $D$ on del Pezzo surfaces, where $\hat{\cal M}=\hat{\cal M}({\rm dP}_n)$ are the generators of the Mori cone.


\section{From data to index formulae}
\label{sec:dphirz_mastform}
The mathematical results presented in this paper have been motivated following a somewhat unorthodox method which might be described as experimental algebraic geometry. The starting point is line bundle cohomology data on various surfaces, produced by algorithmic methods. From this data, we first read off simple piecewise quadratic formulae for cohomology dimensions. These formulae are then put through a process of gradual refinement until they are expressed in terms of intrinsic geometric objects of the underlying surface. In this form, the equations are very suggestive and lead to conjectures for line bundle cohomology on smooth compact complex projective surfaces which we state and prove in the next section. The main purpose of the present section is to describe this ``experimental" approach, focusing on our two main classes of examples, the Hirzebruch and del Pezzo surfaces. This may be of interest to anyone wishing to pursue a similar procedure, for example for a different class of manifolds. The reader mainly interested in the general mathematical results for surfaces can safely skip this section and move on to Section~\ref{sec:genthms}.


\subsection{Outline of approach}
We have already mentioned Refs.~\cite{Constantin:2018hvl,Larfors:2019sie}, where piecewise cubic formulae for line bundle cohomology dimensions on certain Calabi-Yau three-folds have been obtained, starting with cohomology data computed by algorithmic methods. These results suggest that line bundle cohomology dimensions on surfaces, $\surf$, can be described by formulae which are piecewise quadratic in the integers $k_i$ which label the line bundles. This expectation, which we will confirm for our examples, as well as for larger classes of surfaces, is the starting point of our discussion. We will then gradually refine the piecewise quadratic equations in $k_i$ and attempt to re-write them in terms of intrinsic geometric objects, in our quest to uncover the mathematical origin of these equations.\\[2mm]
Let us first recall that line bundles $\lb=\mc{O}_\surf(D) \rightarrow S$ on surfaces have three cohomology dimensions, $h^q(\surf,\lb)$, where $q=0,1,2$. However, Serre duality and the index theorem provide two relations
\be
h^2(\surf,\lb) = h^0(\surf,\dual{\lb} \otimes K_\surf) \,, \quad h^1(\surf,\lb) = h^0(\surf,\lb) + h^2(\surf,\lb) - \ind(\surf,\lb) \,,
\ee
between those three quantities. Here, $K_\surf$ is the canonical bundle\footnote{We will freely write the same symbol $K_\surf$ for the canonical divisor. It will be clear from context which is in use.} of the surface $\surf$ and the index, $\ind(\surf,\lb)$, of the line bundle can be easily computed from the Riemann-Roch formula as
\begin{equation}
 \ind(\surf,\mc{O}_\surf(D))=\ind(\surf,\mc{O}_\surf)+\frac{1}{2}D\cdot (D-K_\surf)\,,\quad  \ind(\surf,\mc{O}_\surf)=\frac{1}{12}(K_\surf^2+\chi(S))\,, \label{RR}
\end{equation} 
where $\chi(S)$ is the Euler characteristic. Written in terms of the line bundle integers $k_i$ the index is a quadratic polynomial. The upshot of this discussion is that knowledge of all zeroth cohomology dimensions $h^0(\surf,\lb)$ for all line bundles $\lb$ determines all other cohomology dimensions.  Moreover, piecewise quadratic formulae for $h^0$ directly translate into piecewise quadratic formulae for $h^1$ and $h^2$ via the above relations. For this reason, we will focus on the zeroth cohomology dimension in the following.\\[2mm]
The basic steps of our approach are as follows.
\begin{itemize}
\item For the complex surface $\surf$, we generate cohomology data, $({\bf k},h^0(\surf,\mc{O}_\surf({\bf k})))$, for a range of integer vectors ${\bf k}$, typically taken from a box with $|k_i|\leq k_{\rm max}$, computed using suitable algorithmic methods.
\item From this data, we extract conjectures for piecewise quadratic formulae for $h^0(\surf,\mc{O}_\surf({\bf k}))$ as a function of ${\bf k}$. Typically this is done by first identifying the regions in the Picard lattice for which the behaviour is quadratic. The experience from Calabi-Yau three folds suggests that these regions are frequently - but not always - cones. For each such region, we then fit a quadratic polynomial in ${\bf k}$ to the data.  In a companion paper~\cite{ml} we explain how this process can be facilitated by methods from machine learning.
\item Next, we attempt to re-write the piecewise quadratic formulae in $k_i$ in a basis-independent way, using intrinsic geometric objects of the underlying surface. For the surfaces we treat, one finds that the relevant objects are the effective (Mori) cone ${\moricn}(S)$ and its list of generators $\hat{\moricn}(S)$, the nef cone ${\nefcn}(S)$ and its list of generators $\hat{\nefcn}(S)$, the irreducible negative self-intersection curves and the intersection form $(D,D')\rightarrow D\cdot D'$ on $\surf$.
\item Finally, we would like to bring the formula into a compact, manageable form. This is particularly relevant for surfaces with high Picard number, where the number of regions in the Picard lattice and, hence, the number of case distinctions required can be large. It turns out that such a compact form can indeed be found using the formula for the index. In this final form, our empirical results are quite suggestive and point to more general mathematical statements which we formulate and prove in Section~\ref{sec:genthms}.
\end{itemize}
The above programme will be carried out for two main classes of surfaces, the Hirzebruch and del Pezzo surfaces, and we now briefly review their basic properties.


\subsection{Basic properties of Hirzebruch surfaces}\label{sec:Hirze}
The Hirzebruch surfaces $\mbb{F}_n$ are indexed by a non-negative integer $n$. They correspond to different fibrations of a $\mathbb{P}^1$ fibre over a $\mathbb{P}^1$ base, with $n$ characterising the twisting. Their non-zero Hodge numbers are $h^{0,0}(\mbb{F}_n)=h^{2,2}(\mbb{F}_n)=1$ and $h^{1,1}(\mbb{F}_n)=2$ and, hence, the Picard number is two for any $n$. The Picard lattice is spanned by the classes of the two projective lines that form the fibre bundle. There exist toric and complete intersection representations of the Hirzebruch surfaces, and details of these have been relegated to Appendix~\ref{app:hirz_surf}. Here, we focus on a few basic properties, relevant to our discussion, which can, for example, be deduced from the toric description. We write $D_1$ for the divisor\footnote{For most of the paper, the term ``divisor" is used as a short-hand for ``divisor class". Whenever the distinction between divisor and divisor class becomes relevant we will state this explicitly.} corresponding to the $\mbb{P}^1$ base of the fibre bundle, and $D_2$ for the divisor corresponding to the $\mbb{P}^1$ fibre. Then, the intersection form of $\mbb{F}_n$ is determined by
\be\label{Hirzeisec}
D_1^2=-n\,,\quad D_1\cdot D_2=1\,,\quad D_2^2=0\, .
\ee
Divisors on $\mbb{F}_n$ can be written as linear combinations $D=k_1D_1+k_2D_2$ with $k_1,k_2\in\mathbb{Z}$ and line bundles $\mc{O}_{\mbb{F}_n}({\bf k})=\mc{O}_{\mbb{F}_n}(k_1D_1+k_2D_2)$ are parametrised by two-dimensional integer vectors ${\bf k}=(k_1,k_2)$.

The two divisors $D_1$ and $D_2$ also correspond to the generators of the Mori cone ${\moricn}(\mbb{F}_n)$, that is,
\be
\hat{\moricn}(\mbb{F}_n)=\{ D_1, D_2\} \,.
\ee
This means effective divisors are of the form $D=k_1D_1+k_2D_2$ with $k_1,k_2 \in \mbb{Z}_{\geq0}$. In addition, from the above intersection form, the dual nef cone has generators
\be
\hat{\nefcn}(\mbb{F}_n)=\{D_2, D_1 + nD_2\} \, . \label{nefFn}
\ee
The anti-canonical divisor $-K_{\mbb{F}_n}$ of the Hirzebruch surface $\mbb{F}_n$ is given by
\be
-K_{\mbb{F}_n} = 2D_1 + (n+2)D_2 \, .
\ee
Inserting into the Riemann-Roch formula~\eqref{RR} an arbitrary  divisor $D=k_1D_1+k_2D_2$, the above expression for the anti-canonical divisor and the result $\ind(\mbb{F}_n,\mc{O}_{\mbb{F}_n})=1$ (which follows from the second equation~\eqref{RR} with $K_{\mbb{F}_n}^2=8$ and $\chi(\mbb{F}_n)=4$) gives
\be \label{indFn}
\ind(\mbb{F}_n,\mc{O}_{\mbb{F}_n}(D)) = 1 + \frac{1}{2}D\cdot(D-K_{\mbb{F}_n}) = 1+k_1+k_2+k_1k_2-\frac{1}{2}nk_1-\frac{1}{2}nk_1^2 \, .
\ee
There are several algorithmic methods to compute line bundle cohomology on Hirzebruch surfaces.
\begin{itemize}
\item Using the toric realisation of the Hirzebruch surfaces (see Appendix~\ref{app:hirz_surf}), the dimension of the zeroth cohomology, which is all we require for our discussion, can be computed from the weight system.
\item Hirzebruch surfaces can also be constructed as complete intersections in products of projective spaces  (see Appendix~\ref{app:hirz_surf}) and the techniques described in Refs.~\cite{Anderson:2007nc,Gray:2007yq,Anderson:2008uw,He:2009wi,Anderson:2009mh}, based on the Bott-Borel-Weil formalism and spectral sequences, can be applied.
\item Finally, again using the toric realisation, we can use the methods to calculate line bundle cohomology on toric spaces developed in Refs.~\cite{CohomOfLineBundles:Algorithm,CohomOfLineBundles:Proof,Jow:2011}.
\end{itemize}


\subsection{Basic properties of del Pezzo surfaces}\label{sec:dPprop}
A del Pezzo surface is isomorphic to either $\mathbb{P}^1\times \mathbb{P}^1$ or to a blow-up of the complex projective plane $\mathbb{P}^2$ in $n\in\{0,1,\ldots ,8\}$ generic points. The cases of $\mathbb{P}^1\times \mathbb{P}^1$ and $\mathbb{P}^2$ are trivial for our purposes since the result directly follows from Bott's formula for line bundle cohomology on projective spaces. For this reason we focus on del Pezzo surfaces which correspond to a blow-up of $\mathbb{P}^2$ in $n\in\{1,\ldots ,8\}$ generic points, and we denote these surfaces by $\dps{n}$.

Del Pezzo surfaces $\dps{n}$ can be realised as complete intersections in products of projective spaces and, for $n=1,2,3$, as toric spaces. Details of this are provided in Appendix~\ref{app:dp_surf}. Here, we merely collect the information essential to our discussion.

The non-zero Hodge numbers of del Pezzo surfaces are $h^{0,0}(\dps{n})=h^{2,2}(\dps{n})=1$ and $h^{1,1}(\dps{n})=n+1$, so the rank of the Picard lattice is $n+1$. A basis for the Picard lattice is given by $(l,\ecl_1,\ldots ,\ecl_n)$, where $\hcl$ is the hyperplane class of $\mathbb{P}^2$ and $\ecl_i$, where $i=1,\ldots ,n$, are the exceptional classes, related to the blow-ups. Relative to this basis, the intersection form is defined by the relations
\begin{equation}
 l^2=1\,,\quad l\cdot \ecl_i=0\,,\quad \ecl_i\cdot \ecl_j=-\delta_{ij}\, .
\end{equation} 
A general divisor is written as
\be
D = k_0 \hcl + \sum_{i=1}^n k_i \ecl_i \,.
\label{eq:dp_bas_exp}
\ee
with $k_0,k_i\in\mathbb{Z}$ and, hence, line bundles $\mc{O}_{\dps{n}}({\bf k})=\mc{O}_{\dps{n}}(k_0\hcl+k_1\ecl_1+\cdots +k_n\ecl_n)$ are labelled by $(n+1)$-dimensional integer vectors ${\bf k}=(k_0,k_1,\ldots ,k_n)$. For another divisor $D'$ with components ${\bf k}'=(k_0',k_1',\ldots ,k_n')$ the intersection form can also be written as
\begin{equation}
 D\cdot D'={\bf k}^TG\,{\bf k}=:\langle{\bf k},{\bf k}'\rangle\,,\quad G={\rm diag}(1,-1,\ldots ,-1)\, . \label{isecG}
\end{equation}
The lists of Mori and nef cone generators are denoted by $\hat{\moricn}=\hat{\moricn}(\dps{n})$ and $\hat{\nefcn}=\hat{\nefcn}(\dps{n})$ respectively, and they are thought of as containing the actual divisors or their coordinates vectors relative to the basis $(l,\ecl_1,\ldots ,\ecl_n)$, depending on context. In the latter form, they are explicitly provided in Appendix~\ref{app:dp_surf}. 

The anti-canonical class of $\dps{n}$ is given by
\be
-K_{\dps{n}} = 3\hcl - \sum_{i=1}^n \ecl_i \, . \label{KdP}
\ee
Inserting into the Riemann-Roch formula a general divisor~\eqref{eq:dp_bas_exp}, the above expression for the anti-canonical class and the result $\ind(\dps{n},\mc{O}_{\dps{n}})=1$ (which follows from the second Eq.~\eqref{RR} with $K_{\dps{n}}^2=9-n$ and $\chi(\dps{n})=3+n$) gives
\be\label{inddP}
\ind(\dps{n},\mc{O}_{\dps{n}}(D))= 1 + \frac{1}{2}D\cdot(D-K_{\dps{n}}) = 1+\frac{1}{2}k_0(k_0+3)+\frac{1}{2}\sum_{i=1}^nk_i(1-k_i) \,.
\ee
As for Hirzebruch surfaces, there are several algorithmic methods available to calculate line bundle cohomology on del Pezzo surfaces.
\begin{itemize}
\item For the del Pezzo surfaces $\dps{n}$, with $n=1,2,3$ which have a toric realisation (see Appendix~\ref{app:dp_surf}), we can compute $h^0$ either from the toric weight system or via the methods for line bundle cohomology on toric spaces from Refs.~\cite{CohomOfLineBundles:Algorithm,CohomOfLineBundles:Proof,Jow:2011}.
\item All del Pezzo surfaces have a realisation as complete intersections in products of projective spaces (see Appendix~\ref{app:dp_surf}), so the methods of Refs.~\cite{Anderson:2007nc,Gray:2007yq,Anderson:2008uw,He:2009wi,Anderson:2009mh} can be applied.
\item It was noticed in Appendix~B of Ref.~\cite{Blumenhagen:2008zz} that line bundle cohomology on del Pezzo surfaces can be computed by counting certain polynomials on $\mathbb{P}^2$ and we use  a computational implementation of this method.
\end{itemize}
There is one further result which relates line bundle cohomology of del Pezzo surfaces $\dps{n}$ for different $n$ which will be helpful in the following. The cohomology dimension of a divisor $D$ on $\dps{n+1}$ with no component along the exceptional class $\ecl_{n+1}$ is the same as the cohomology dimension of $D$ seen as a divisor on $\dps{n}$, that is,
\begin{equation}
 h^0(\dps{n+1},\mc{O}_{\dps{n+1}}(D))=h^0(\dps{n},\mc{O}_{\dps{n}}(D))\quad\mbox{for}\quad D=k_0\hcl+\sum_{i=1}^nk_i\ecl_i\, . \label{dPnrel}
\end{equation} 
This can be shown, for example, by thinking about blowing-down one exceptional divisor or by using the relation between del Pezzo cohomology and polynomials on $\mathbb{P}^2$ from Ref.~\cite{Blumenhagen:2008zz}.


\subsection{Warm up: line bundle cohomology on Hirzebruch surfaces}
The complexity of our task clearly increases with the Picard number of the surface. Hirzebruch surfaces, which all have Picard number two, therefore provide a simple setting for an initial exploration. In particular, it is possible to plot cohomology data and identify the regions in the Picard lattice by ``eyeballing". Once the regions are known, the quadratic polynomials can be fixed by a simple fit to a number of points in each region. This results in a piecewise quadratic formula which represents the first step on our path from data to general mathematical statements. Of course this piecewise quadratic formula should then be checked against all available cohomology data.\\[2mm]
We recall that line bundles on Hirzebruch surfaces $\mbb{F}_n$ are labelled by a two-dimensional integer vector ${\bf k}=(k_1,k_2)$, relative to the divisor basis introduced in Section~\ref{sec:Hirze}. In Figure~\ref{fig:hirz_f0f1f2_cohomplot} we plot the zeroth cohomology dimension, $h^0(\mbb{F}_n,\mc{O}_{\mbb{F}_n}({\bf k}))$, as a function of $(k_1,k_2)$, for the first three Hirzebruch surfaces $\mbb{F}_0$, $\mbb{F}_1$ and $\mbb{F}_2$. For a better visualisation, we have joined the discrete data points into a surface. \begin{figure}[h]
\begin{minipage}{0.32\linewidth}
\begin{center}
\begin{picture}(160,180)
\put(0,0){\includegraphics[scale=.26]{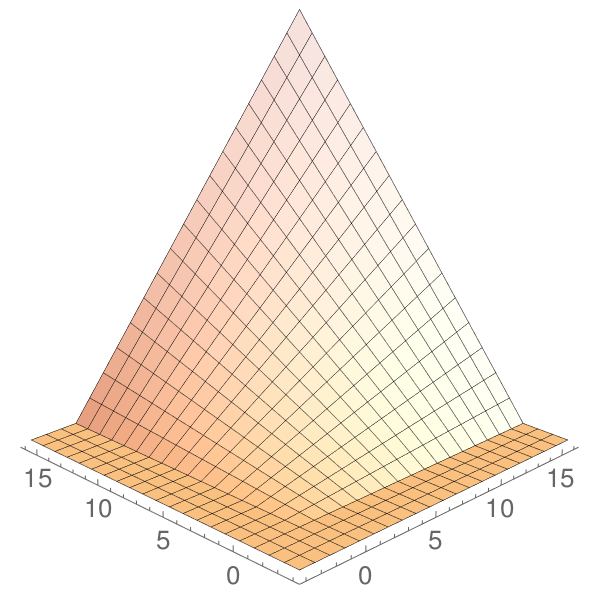}}
\put(73,160){$\mbb{F}_0$}
\put(35,-1){\scriptsize{$k_2$}}
\put(113,-1){\scriptsize{$k_1$}}
\put(140,60){\scriptsize{$h^0$}}
\end{picture}
\end{center}
\end{minipage}
\begin{minipage}{0.32\linewidth}
\begin{center}
\begin{picture}(160,180)
\put(0,0){\includegraphics[scale=.26]{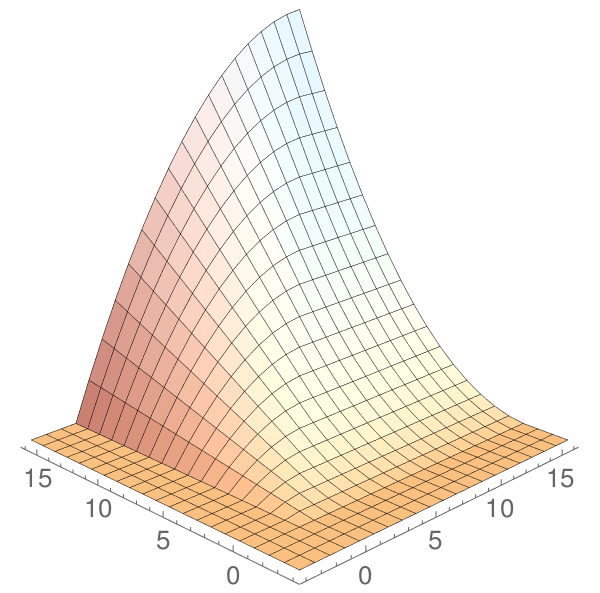}}
\put(73,160){$\mbb{F}_1$}
\put(35,-1){\scriptsize{$k_2$}}
\put(113,-1){\scriptsize{$k_1$}}
\put(140,60){\scriptsize{$h^0$}}
\end{picture}
\end{center}
\end{minipage}
\begin{minipage}{0.32\linewidth}
\begin{center}
\begin{picture}(160,180)
\put(0,0){\includegraphics[scale=.26]{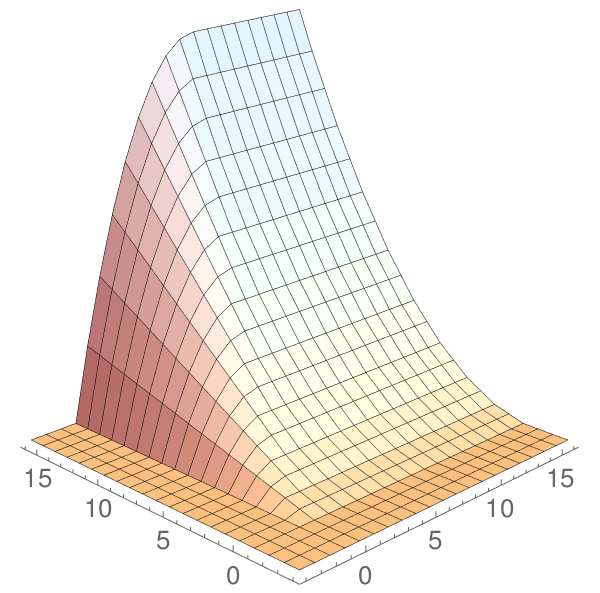}}
\put(73,160){$\mbb{F}_2$}
\put(35,-1){\scriptsize{$k_2$}}
\put(113,-1){\scriptsize{$k_1$}}
\put(140,60){\scriptsize{$h^0$}}
\end{picture}
\end{center}
\end{minipage}
\caption{\sf Zeroth cohomology $h^0(\mbb{F}_n,\mc{O}_{\mbb{F}_n}(k_1D_1+k_2D_2))$ as a function of $(k_1,k_2)$ for the Hirzebruch surfaces $\mbb{F}_0$, $\mbb{F}_1$, $\mbb{F}_2$. For clarity, we have joined the discrete data points into a surface.}
\label{fig:hirz_f0f1f2_cohomplot}
\end{figure}
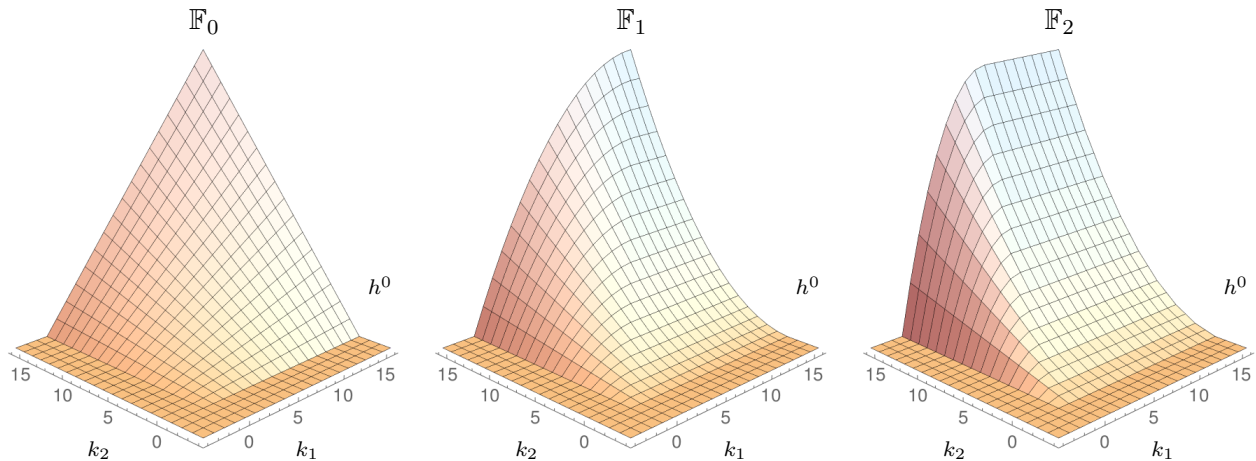
We recall from Section~\ref{sec:Hirze} that the cone of effective divisors, that is, the region where $h^0(\mbb{F}_n,\mc{O}_{\mbb{F}_n}({\bf k}))>0$, is characterised by $k_1\geq 0$ and $k_2\geq 0$. This is consistent with the plots in Figure~\ref{fig:hirz_f0f1f2_cohomplot} which indicate a non-zero cohomology precisely in the positive quadrant.

What is the structure in the positive quadrant? The obvious feature in the plots for $\mbb{F}_1$ and $\mbb{F}_2$ in Figure~\ref{fig:hirz_f0f1f2_cohomplot} is the presence of two regions which we expect require two different quadratic polynomials. In fact, it can be checked that this structure persists for all Hirzebruch surfaces with $n \geq 1$ and that the two regions are separated by the hyperplane through the origin, described by the equation
\be
k_2 = n k_1 \, . \label{Hirzebound}
\ee
Now that we have identified the regions we can attempt polynomial fits. The region $k_2 \geq n k_1$ corresponds to the easy case. Here, the relevant polynomial for all Hirzebruch surfaces is simply the index, given in Eq.~\eqref{indFn}. The other region, where $k_2 < n k_1$ and which exists for all Hirzebruch surfaces with $n\geq 1$, is more problematic. A polynomial fit to the data in this region succeeds for $\mbb{F}_1$ but it fails for $\mbb{F}_2$ and all higher Hirzebruch surfaces. The problem becomes apparent when we plot the cohomology data for $\mbb{F}_3$, $\mbb{F}_4$, and $\mbb{F}_5$ as we have done in Figure~\ref{fig:hirz_f3f4f5_cohomplot}.
\begin{figure}[h]
\begin{minipage}{0.32\linewidth}
\begin{center}
\begin{picture}(160,180)
\put(0,0){\includegraphics[scale=.26]{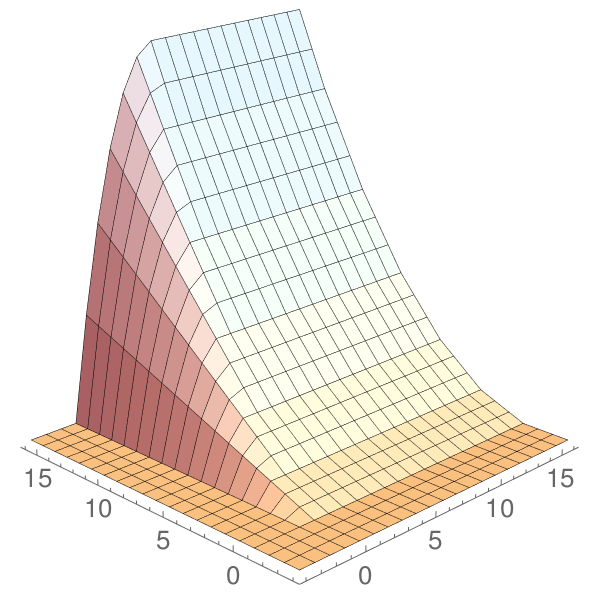}}
\put(73,160){$\mbb{F}_3$}
\put(35,-1){\scriptsize{$k_2$}}
\put(113,-1){\scriptsize{$k_1$}}
\put(140,60){\scriptsize{$h^0$}}
\end{picture}
\end{center}
\end{minipage}
\begin{minipage}{0.32\linewidth}
\begin{center}
\begin{picture}(160,180)
\put(0,0){\includegraphics[scale=.26]{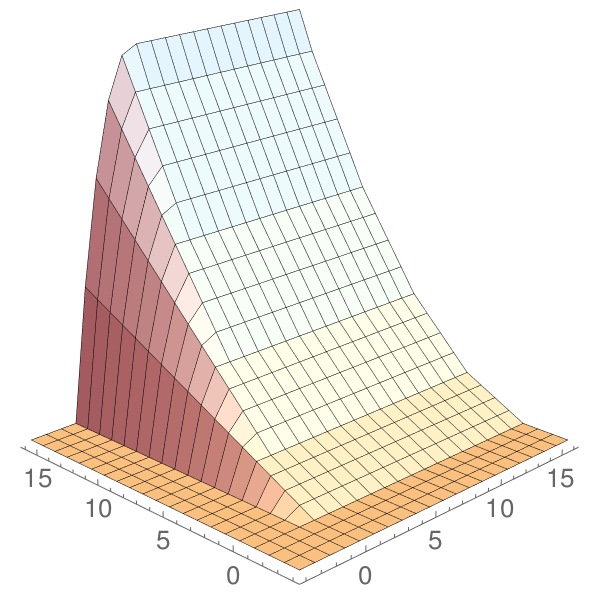}}
\put(73,160){$\mbb{F}_4$}
\put(35,-1){\scriptsize{$k_2$}}
\put(113,-1){\scriptsize{$k_1$}}
\put(140,60){\scriptsize{$h^0$}}
\end{picture}
\end{center}
\end{minipage}
\begin{minipage}{0.32\linewidth}
\begin{center}
\begin{picture}(160,180)
\put(0,0){\includegraphics[scale=.26]{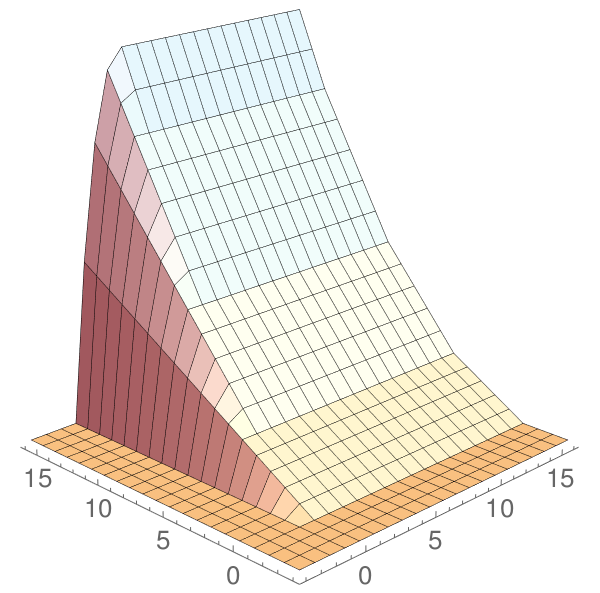}}
\put(73,160){$\mbb{F}_5$}
\put(35,-1){\scriptsize{$k_2$}}
\put(113,-1){\scriptsize{$k_1$}}
\put(140,60){\scriptsize{$h^0$}}
\end{picture}
\end{center}
\end{minipage}
\caption{\sf Zeroth cohomology $h^0(\mbb{F}_n,\mc{O}_{\mbb{F}_n}(k_1D_1+k_2D_2))$ on a region of the Picard lattice for the Hirzebruch surfaces $\mbb{F}_3$, $\mbb{F}_4$, $\mbb{F}_5$. This function is lattice-valued; we have joined the lattice data points into a surface for clarity.}
\label{fig:hirz_f3f4f5_cohomplot}
\end{figure}
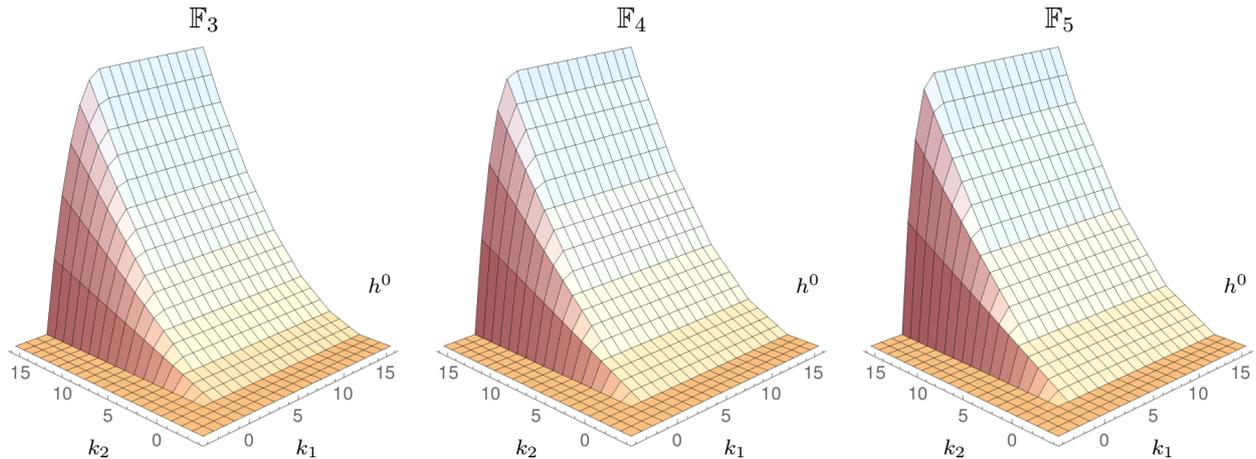
The cohomology values in the region $k_2 < n k_1$ show a mod $n$ structure which can of course not be captured by a single quadratic polynomial. A similar structure has been observed in some of the examples studied in Ref.~\cite{Klaewer:2018sfl}.  However, it is not too difficult to account for this mod $n$ behaviour by including a ceiling function. This leads to the following conjecture
\be\label{h0Hirze}
\begin{array}{rll}
h^0\left(\mbb{F}_n,\mc{O}_{\mbb{F}_n}({\bf k})\right) &=& \left\{
  \begin{array}{lll}
    \ind(\mc{O}_{\mbb{F}_n}({\bf k}))
    & \textrm{if $k_1,k_2\geq0$ and $k_2 \geq nk_1$}&\mbox{(Region 1)}  \\[4pt]
    \displaystyle \frac{1}{2}(1-c_2)(2+2k_2+n c_2)
    & \textrm{if $k_1,k_2\geq0$ and  $k_2 < nk_1$} &\mbox{(Region 2)}  \\[8pt]
    0
    & \textrm{otherwise}&\mbox{(Region 3)} 
  \end{array}
  \right.\\\\
 \text{where } c_2&=& \ceil{ \frac{-k_2}{n} }\, 
\end{array} 
\ee
for the zeroth cohomology dimension on any Hirzebruch surface $\mbb{F}_n$. The explicit expression for the index can be found in Eq.~\eqref{indFn}.  This result lends further support to the conjecture that line bundle cohomology on complex surfaces is described by equations which are (basically) piecewise quadratic. 
However, the appearance of the ceiling function which accounts for the mod $n$ behaviour is new and, as we will see, points to an important feature of the more general formula we are seeking.\\[2mm]
This completes the first two parts of our programme. The remaining tasks are, firstly, to write Eq.~\eqref{h0Hirze} in a basis-independent way and, secondly, to find a compact form, in terms of natural geometric objects. The current example of the Hirzebruch surfaces is simple enough to carry out both steps at once.

Consider first the boundary between Region 1 and Region 2. A quick glance at the intersection rules~\eqref{Hirzeisec} shows that, for $D=k_1D_1+k_2D_2$, we have $D \cdot D_1 = -nk_1+k_2$. Hence, these two regions can be characterised by saying that $D$ needs to be in the effective cone, while, in addition, we require that $D\cdot D_1\geq 0$ for Region 1 and $D\cdot D_1<0$ for Region 2. Finally, Region 3 is characterised by saying that $D$ is not in the effective cone. Since we can think of $D_1=C$ as the unique irreducible divisor with negative self-intersection, this provides a natural basis-independent formulation of the regions.

What about the polynomial expressions in those regions?  In Region 1, the cohomology is already described by an intrinsic geometrical object, the index, but it is not immediately obvious how to proceed in Region 2.  A useful observation is that the cohomology dimension in Region 2 does not depend on $k_1$. This means that a projection exists: one can relate a cohomology result in Region 2 to one in Region 1 by a lattice projection along the negative $k_1$ direction. Since cohomology dimensions in Region 1 are described by the index we conclude that the same must be true for Region 2, however, the argument of the index is now a different, projected divisor. It turns out that the required lattice projection is
\be
D \to \tilde{D}=D - \ceil{\frac{k_2-nk_1}{-n}}D_1 = D - \ceil{\frac{D \cdot C}{-n}}C \,, \label{shiftFn}
\ee
with $C=D_1$. Note that the expression on the right-hand-side is written entirely in terms of natural geometric objects  of the Hirzebruch surface. With this projection, we can re-write the cohomology Eq.~\eqref{h0Hirze} as
\be
h^0\left(\mbb{F}_n,\mc{O}_{\mbb{F}_n}(D)\right) = \left.
  \begin{cases}
   \displaystyle{ \ind\left(D - \theta\left(-D\cdot C\right)\ceil{\frac{-D\cdot C}{-n}} C\right) }
    & \textrm{if } D\in{\moricn}(\mbb{F}_n) \\
    0
    & \textrm{if } D\notin{\moricn}(\mbb{F}_n)
  \end{cases}
  \right.
\label{eq:hirzsurf_indmastform}
\ee
where $C=D_1$ is the unique irreducible divisor with negative self-intersection and $\theta$ is the Heaviside step function, defined by $\theta(x) = 1$ for $x > 0$ and $\theta(x) = 0$ otherwise.

The above result, albeit in a very suggestive mathematical form, is still a conjecture since it has been extracted from a finite amount of cohomology data. In fact, Hirzebruch surfaces are sufficiently simple so that a direct proof can be found, for example from the weight system of the associated toric diagrams. We will not pursue this explicitly since, as we will see, Eq.~\eqref{eq:hirzsurf_indmastform} will direct us towards the general mathematical results in Section~\ref{sec:genthms}, of which it will turn out to be a special case.


\subsection{Regions and polynomials for del Pezzo surfaces}
We would now like to repeat the process from data to a concise mathematical formula for the case of del Pezzo surfaces. This is of course considerably more complicated, given that Picard numbers reach up to nine. In fact, writing down piecewise polynomial formulae explicitly becomes impractical for higher Picard numbers as the number of regions and, hence, case distinctions required, increases considerably. This is of course one of the reasons we are seeking concise mathematical formulae for cohomology dimensions. 

We will tackle the problem of finding concise formulae in the next two sub-sections, but presently we would like to explain how to extract piecewise polynomial formulae from data. In practice we begin with the lower del Pezzo surfaces, $\dps{n}$, and gradually increase $n$. Eq.~\eqref{dPnrel} tell us that the structure found for $\dps{n}$ is preserved as we move on to $\dps{n+1}$.  This means that at each stage we have some partial information available from the del Pezzo surfaces with lower $n$. For example, when determining some hyperplane boundaries, we already know some entries in their normal vectors.

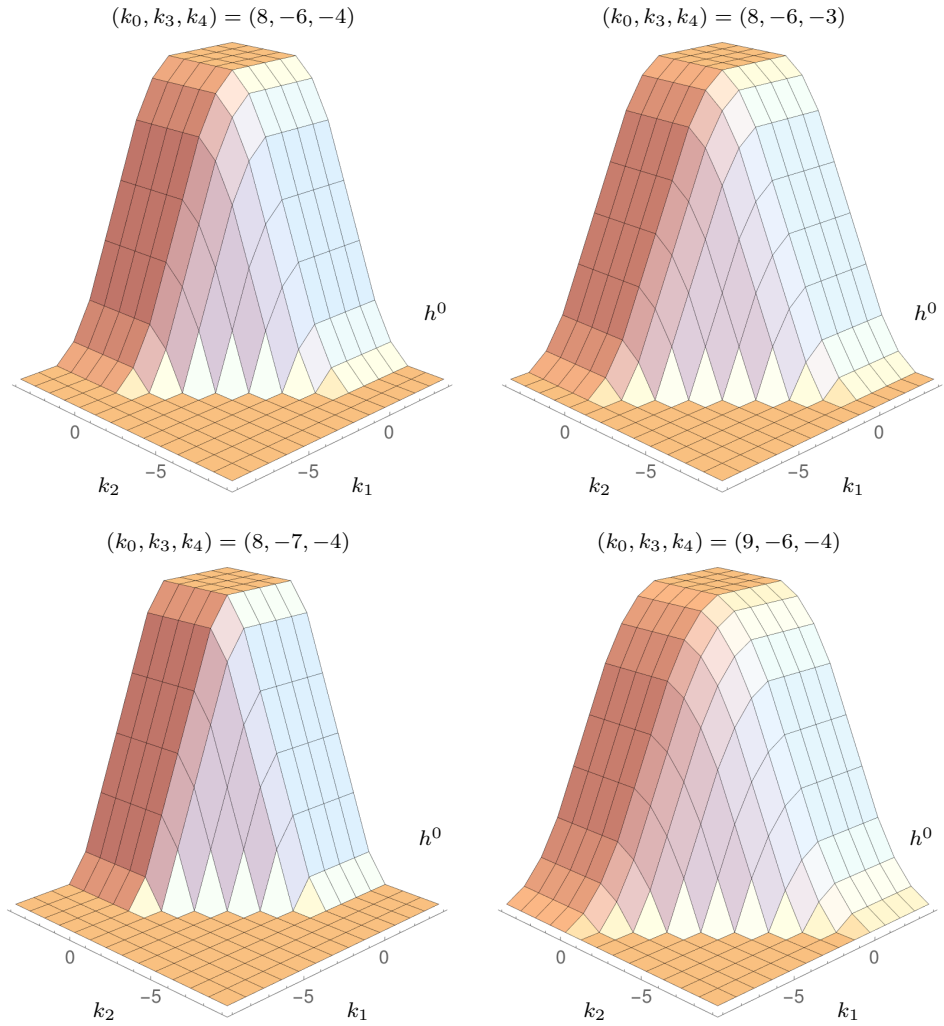
\begin{figure}[ht]
\vspace{.1cm}
\begin{picture}(180,180)
\put(0,0){\includegraphics[scale=.22]{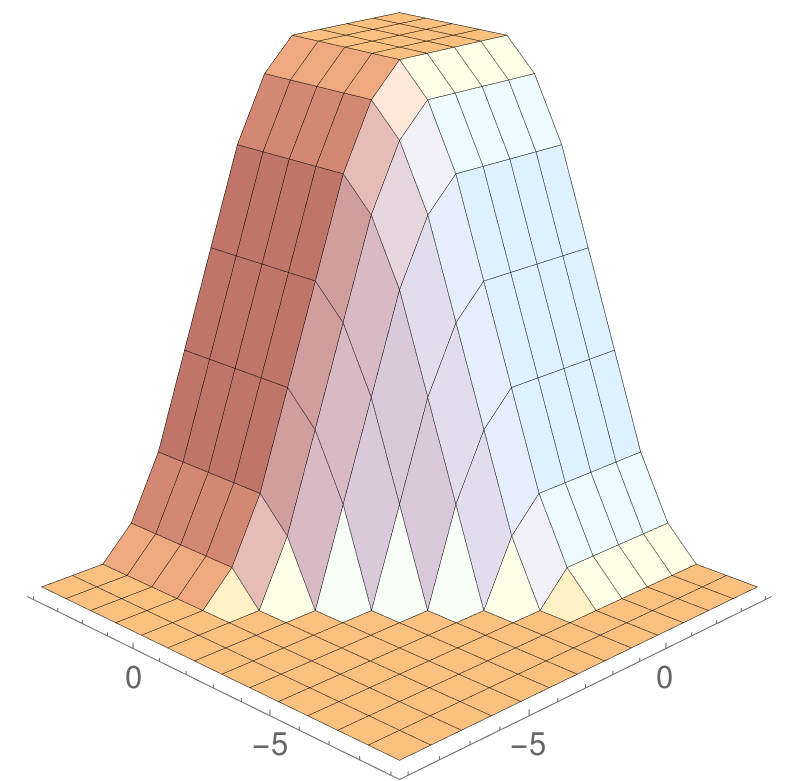}}
\put(43,177){\scriptsize{$(k_0,k_3,k_4)=(8,-6,-4)$}}
\put(38,0){\scriptsize{$k_2$}}
\put(133,0){\scriptsize{$k_1$}}
\put(160,65){\scriptsize{$h^0$}}
\end{picture}
\begin{picture}(180,180)
\put(0,0){\includegraphics[scale=.22]{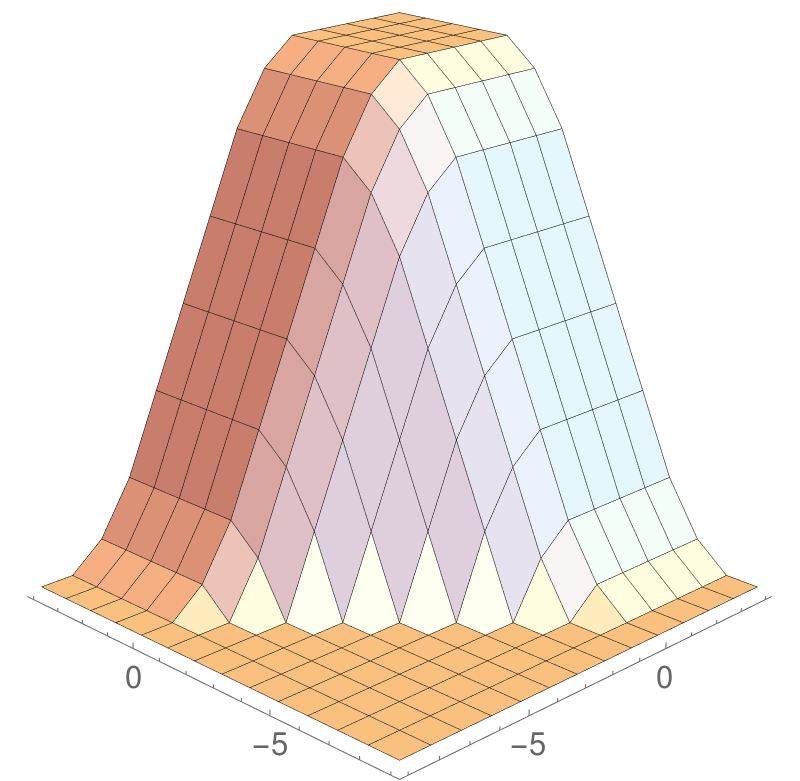}}
\put(43,177){\scriptsize{$(k_0,k_3,k_4)=(8,-6,-3)$}}
\put(38,0){\scriptsize{$k_2$}}
\put(133,0){\scriptsize{$k_1$}}
\put(160,65){\scriptsize{$h^0$}}
\end{picture} \\ \vspace{.6cm}
\begin{picture}(180,180)
\put(0,0){\includegraphics[scale=.22]{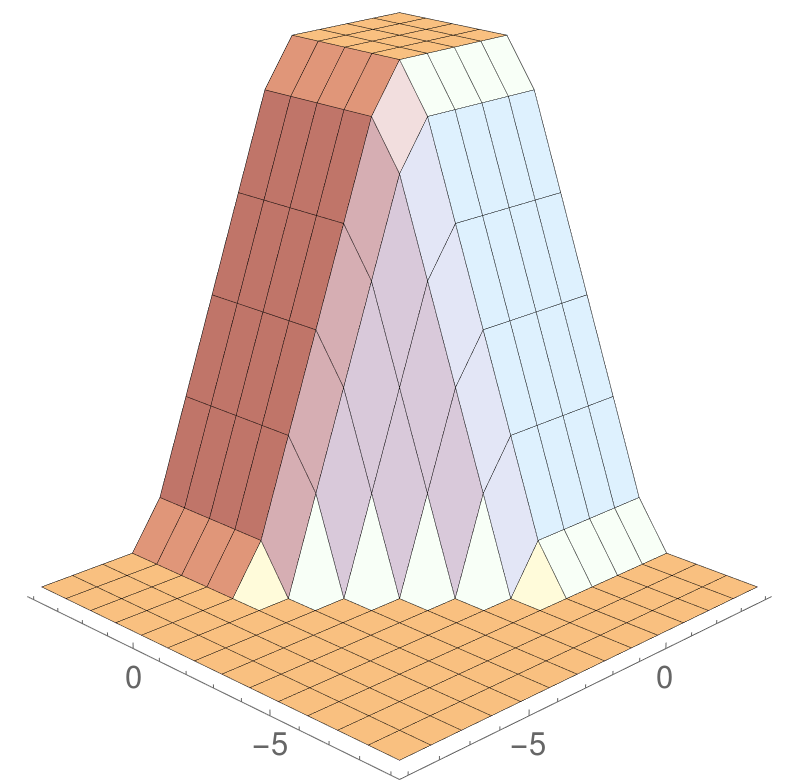}}
\put(43,177){\scriptsize{$(k_0,k_3,k_4)=(8,-7,-4)$}}
\put(38,0){\scriptsize{$k_2$}}
\put(133,0){\scriptsize{$k_1$}}
\put(160,65){\scriptsize{$h^0$}}
\end{picture}
\begin{picture}(180,180)
\put(0,0){\includegraphics[scale=.22]{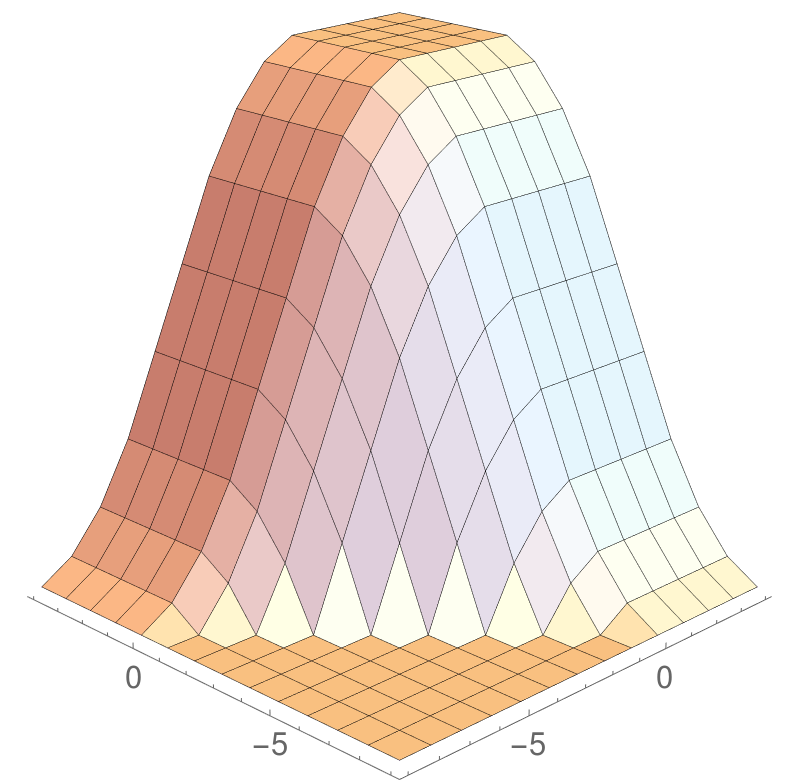}}
\put(43,177){\scriptsize{$(k_0,k_3,k_4)=(9,-6,-4)$}}
\put(38,0){\scriptsize{$k_2$}}
\put(133,0){\scriptsize{$k_1$}}
\put(160,65){\scriptsize{$h^0$}}
\end{picture}
\vspace{.1cm}
\caption{\sf Zeroth cohomology $h^0(\dps{4},\mc{O}_{\dps{4}}(k_0\hcl +k_1\ecl_1 + \ldots + k_4\ecl_4))$ as a function of  $(k_1, k_2)$ in the range $k_1,k_2=-9,\ldots ,4$, for four sets of fixed values for $(k_0,k_3,k_4)$. The lattice data points have been joined into surfaces for clarity.}
\label{fig:dp4_slices}
\end{figure}

The first step is to find the region boundaries and regions. Once this has been accomplished it is easy to find the quadratic polynomial in each region by a simple fit to the data. For the cases $\dps{1}$ and $\dps{2}$  we can proceed as for the Hirzebruch surfaces, that is, by simply reading the boundaries off from a plot.  Even for higher Picard numbers it is possible to proceed by ``eyeballing" if we focus on two-dimensional slices in the Picard lattice and combine the information from various such slices. We recall that line bundles on $\dps{n}$ are labelled by an $(n+1)$-dimensional integer vector ${\bf k}=(k_0,k_1,\ldots ,k_n)$, with the associated divisor given in Eq.~\eqref{eq:dp_bas_exp}.

Let us illustrate the process of finding the region boundaries for the example $\dps{4}$. The Picard number of this space is five, so a two-dimensional slice misses the behaviour of a boundary in the remaining three directions. We begin with slices in the $(k_1,k_2)$ plane, taken for various fixed choices of the remaining coefficients $(k_0,k_3,k_4)$.
In Figure~\ref{fig:dp4_slices} we have plotted four such two-dimensional slices through the Picard lattice of $\dps{4}$ which show several region boundaries. As an example, we focus on the diagonal boundary which separates the regions of zero and non-zero cohomology.

The dependence on $k_1$ and $k_2$ is immediately clear: this boundary must be of the form $k_1+k_2 + \ldots = 0$. To find the dependence on the remaining $k_i$, we compare slices. Take as a starting point the slice with $(k_0,k_3,k_4)=(8,-6,-4)$ (upper left plot). When we increase $k_4$ (upper right plot), the boundary advances by one unit and when $k_3$ decreases by one (lower left plot), the boundary recedes by one unit. Finally, when we instead increase $k_0$ (lower right plot), the boundary advances by two units. Hence the equation of this boundary is
\be
2k_0+k_1+k_2+k_3+k_4=0 \,.
\ee
As one would expect, four slices were sufficient to determine this boundary. The other boundaries visible in Fig.~\ref{fig:dp4_slices} can be determined in the same way. The polynomials can then be obtained by a fit to the data in each region.  The formulae obtained in this way are of course still conjectures, extracted from a finite amount of cohomology data.\\[2mm]
As already mentioned, the piecewise polynomial formulae for $\dps{n}$ become quite complicated for larger $n$ and have too many case distinctions to be reproduced here. However, it is possible to discuss $\dps{2}$ in a concise way and we do this now for illustration. The Picard number of this surface is three and line bundles are labelled by a three-dimensional integer vector ${\bf k}=(k_0,k_1,k_2)$. Figure~\ref{fig:dp2_numbregs} is a plot of the regions that correspond to distinct polynomials. All region boundaries in this case are hyperplanes through the origin, so the regions themselves are cones with their bases at the origin. 
\begin{figure}
  \includegraphics[scale=0.7]{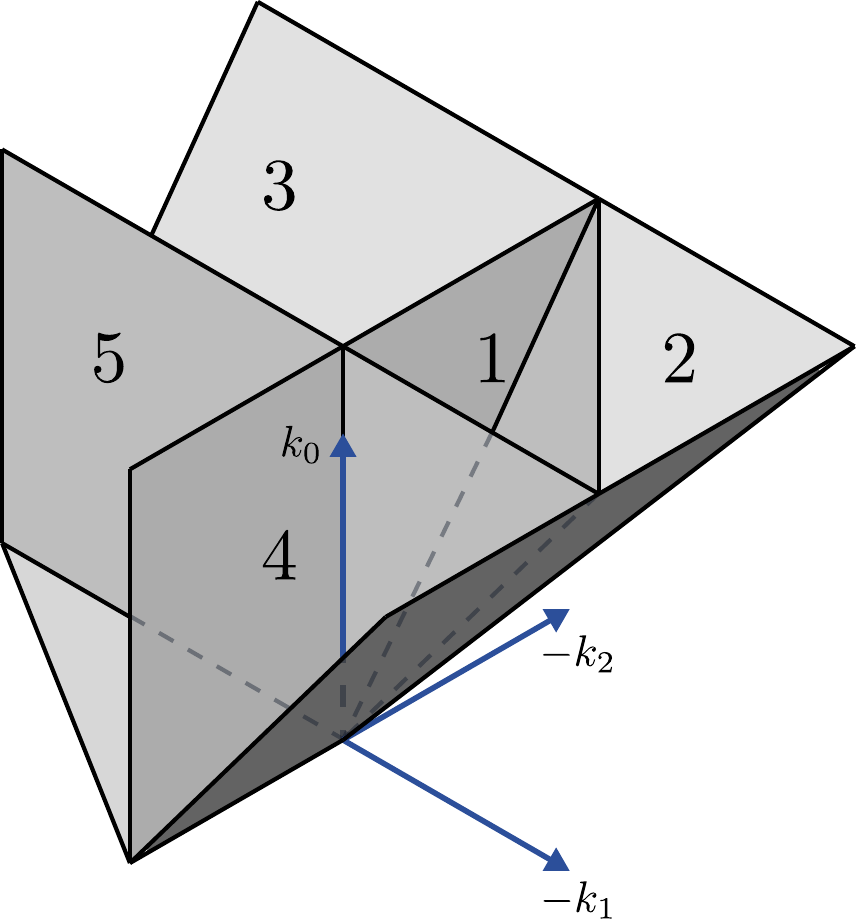}
  \caption{\sf Depiction of the Picard lattice of the del Pezzo surface $\dps{2}$. We do not draw the lattice points to avoid clutter. The five numbered cones are regions where different non-zero polynomials describe the zeroth cohomology; together they make up the Mori cone.}
  \label{fig:dp2_numbregs}
\end{figure}
Evidently, there are five distinct non-zero regions as numbered in the diagram, and they are described by the following sets of inequalities.
\be \label{dP2regions}
\begin{tabular}{c C C C C C C}
Region 1: & -k_1 \geq 0 	& -k_2 \geq 0 	& k_0+k_1+k_2 \geq 0 	& 				& 				& \\
Region 2: & 			&  			& k_0+k_1+k_2 < 0 	& k_0+k_1 \geq 0 	& k_0+k_2 \geq 0 	& \\
Region 3: & -k_1 < 0	& -k_2 \geq 0	& 					& 			 	& k_0+k_2 \geq 0 	& \\
Region 4: & -k_1 \geq 0	& -k_2 < 0	& 					& 			 	& k_0+k_2 \geq 0 	& \\
Region 5: & -k_1 < 0	& -k_2 < 0	& 					& 			 	&  				& k_0 \geq 0 \\
\end{tabular}
\ee
Note that Region 1 is the nef cone and that  the union of the five regions is the Mori cone. The polynomials that describe the zeroth cohomology dimension in each of these five regions are as follows. \\
\be
h^0(\dps{2},\mc{O}_{\dps{2}}({\bf k})) = \left\{
\begin{array}{lll}
\displaystyle{1+\frac{3}{2}k_0+\frac{1}{2}k_0^2+\frac{1}{2}}k_1-\frac{1}{2}k_1^2+\frac{1}{2}k_2-\frac{1}{2}k_2^2
& \textrm{in Region 1}  \vspace{1mm} \\
\displaystyle{1+2k_0+k_0^2+k_1+k_0k_1+k_2+k_0k_2+k_1k_2 \phantom{\displaystyle{\frac{1}{2}}}}
& \textrm{in Region 2} \vspace{1mm}  \\
\displaystyle{1+\frac{3}{2}k_0+\frac{1}{2}k_0^2+\frac{1}{2}k_2-\frac{1}{2}k_2^2}
& \textrm{in Region 3}  \vspace{1mm} \\
\displaystyle{1+\frac{3}{2}k_0+\frac{1}{2}k_0^2+\frac{1}{2}k_1-\frac{1}{2}k_1^2}
& \textrm{in Region 4}  \vspace{1mm} \\
\displaystyle{1+\frac{3}{2}k_0+\frac{1}{2}k_0^2}
& \textrm{in Region 5} \,. \\
\end{array}\right.
\label{eq:dp2_polys}
\ee
The main features of this formula - regions which are cones with bases at the origin and cohomology dimensions described by a quadric in each region - persist for all del Pezzo surfaces. Writing down the analogous formulae for $\dps{n}$ with $n>2$ becomes impractical as the number of regions increases very quickly with $n$. For this reason, we will now  extract a more concise formula which works for all del Pezzo surfaces.\\[2mm]
There is a more sophisticated way to extract regions and polynomials from data by using methods from machine learning. This will be discussed in detail in a companion paper~\cite{ml}.


\subsection{A compact formula for del Pezzo surfaces}
\label{sec:regandpols_to_compform}
The first step towards extracting a concise formula which works for all del Pezzo surfaces is to understand the structure of the hyperplanes which separate the regions. Each such hyperplane is determined by a normal vector ${\bf v}$ such that the boundary is described by the equation $\langle{\bf v},{\bf k}\rangle={\bf v}^TG\,{\bf k}=0$, where $G$ is the intersection matrix in Eq.~\eqref{isecG}. For $\dps{2}$ these normal vectors can be read off from the inequalitites in the first three columns of \eqref{dP2regions}, while the last three columns correspond to the boundaries of the Mori cone. Repeating the exercise for $\dps{3}$ we find the following lists of normal vectors.
\be
\begin{aligned}
\dps{2} : \quad&{\bf v}\in \{ (0,1,0) \,,\, (0,0,1) \,,\, (1,-1,-1) \} \\
\dps{3} : \quad&{\bf v}\in \{ (0,1,0,0) \,,\, (0,0,1,0) \,,\, (0,0,0,1) \,,\, (1,-1,-1,0) \,,\, (1,-1,0,-1) \,,\, (1,0,-1,-1) \} 
\end{aligned}\, .
\ee
For $\dps{4}$, the normal vectors have the same structure as for $\dps{3}$ but with the other obvious permutations included. The same is true for $\dps{5}$, except, additionally, the new vector ${\bf v}=(2,-1,-1,-1,-1,-1)$ appears.

The reader familiar with the properties of del Pezzo surfaces will immediately recognise the above divisors as the generators of the Mori cone or, equivalently, as the exceptional divisors with self-intersection $-1$. We conclude that the bounding hyperplanes for the various regions are described by the Mori cone generators and, it turns out, this holds for all $\dps{n}$ with $n=1,\ldots ,8$. We denote the set of Mori and nef cone generators for $\dps{n}$ by $\hat{\moricn}$ and $\hat{\nefcn}$, respectively,  and their elements are explicitly listed in Appendix~\ref{app:dp_surf}.\\[2mm]
Our next task is to find the quadratic polynomials that describe the cohomology in the various regions of the Picard lattice. To this end, it is instructive to look at the difference between two polynomials $P$ and $\tilde{P}$ in neighbouring regions, separated by a boundary hyperplane with normal vector ${\bf v}$. By inspecting examples, it turns out that the change is always of the form
\be
P - \tilde{P} = \frac{1}{2}q(q+1) \,,
\ee
where $q$ is a linear expression in the $k_i$ with integer coefficients. In fact these integer coefficients follow a pattern too. If the region with polynomial $P$ is characterised by $\langle{\bf k},{\bf v}\rangle \geq 0$ and the one with polynomial $\tilde{P}$ by $\langle{\bf k},{\bf v}\rangle < 0$, where ${\bf v}\in\hat{\moricn}$, then we have
\begin{equation}
 P-\tilde{P}=\frac{1}{2}\langle{\bf k},{\bf v}\rangle(\langle{\bf k},{\bf v}\rangle+1)\, .
\end{equation} 
The above term on the right-hand side appears when we cross a boundary with normal vector ${\bf v}$ into a region with $\langle{\bf k},{\bf v}\rangle< 0$. It is therefore natural to start in the ``simplest" region where $\langle{\bf k},{\bf v}\rangle\geq 0$ for all ${\bf v}\in\hat{\moricn}$. This region is of course precisely the nef cone ${\nefcn}(\dps{n})$, the dual of the Mori cone, where the zeroth cohomology dimension is given by the index. Hence, for the divisor $D=k_0\hcl+k_1\ecl_1+\cdots +k_n\ecl_n$, we have the following formula 
\be
h^0\left(\dps{n},\mc{O}_{\dps{n}}({\bf k})\right) = \left.
  \begin{cases}
  \displaystyle{\ind(\mc{O}_{\dps{n}}({\bf k})) +\frac{1}{2} \sum_{{\bf v}\in\hat{\moricn}}\theta\left(-\langle{\bf k},{\bf v}\rangle\right)}
  \langle{\bf k},{\bf n}\rangle(\langle{\bf k},{\bf n}\rangle+1)
    & \textrm{if }\left\{\begin{array}{l}\langle{\bf k},{\bf w}\rangle\geq 0\\\, \forall\,{\bf w}\in\hat{\nefcn}\end{array}\right. \\
    0
    & \textrm{otherwise}
  \end{cases}
  \right. 
\label{eq:compexp}
\ee
for the zeroth cohomology dimensions, where we recall that $\langle{\bf k},{\bf v}\rangle={\bf k}^TG\,{\bf v}$ is the intersection form in our standard basis, as defined in Eq.~\eqref{isecG}. Further, $\hat{\moricn}$ and $\hat{\nefcn}$ are the sets of Mori and nef cone generators which are explicitly provided in Appendix~\ref{app:dp_surf}.

As for the earlier formula~\eqref{eq:dp2_polys} for $\dps{2}$, this result has been extracted from data and is, hence, still at the level of a conjecture. However, we have validated Eq.~\eqref{eq:compexp} for all $\dps{n}$ by comparing with data in a box with $|k_i|\leq k_{\rm max}\simeq 20$. \\[2mm]
It is easy to convert Eq.~\eqref{eq:compexp} into the basis-independent form
\be\label{h0dP2}
 h^0(\dps{n},\mc{O}_{\dps{n}}(D))=
 \left\{
 \begin{array}{lll}
  \displaystyle{\ind(D) + \frac{1}{2}\sum_{C \in \hat{\moricn}}\theta\left(-D \cdot C\right)(D \cdot C)\left(D \cdot C + 1\right)}
   &\mbox{if }D\in {\moricn}(\dps{n})\\
  0&\mbox{otherwise}
 \end{array}\right.\, .
 \ee 
 Even though this is already quite concise there is an even better way of writing this formula which suggests a natural origin of the additional terms in the above sum.  We will now derive this alternative form.


\subsection{An index formula for del Pezzo surfaces}
\label{sec:dp_fromcomp_toind}
Our guiding principle for a re-formulation of Eq.~\eqref{h0dP2} is the hope that the entire right-hand side can be written as an index, similar to what we have found for Hirzebruch surfaces.  In this context, a key observation is that each additional term $\frac{1}{2}\left(D \cdot C\right)\left(D \cdot C + 1\right)$ in Eq.~\eqref{h0dP2} can be written as a difference between two indices.  Specifically, we have
\be
\begin{aligned}
\ind\left(D+(D \cdot C)C \right) 
&= 1 + \frac{1}{2}\left(D+(D \cdot C)C\right) \cdot \left(D+(D \cdot C)C - K_{\dps{n}} \right) \\
&= \ind\left(D\right) + \frac{1}{2}(D \cdot C)\left[ 2(D \cdot C) + (D \cdot C)C^2 - (C \cdot K_{\dps{n}}) \right] \\
&= \ind\left(D\right) + (D \cdot C)\left[ (D \cdot C) + 1 \right] \,,
\end{aligned}
\label{eq:ind_dif}
\ee
where the first equality follows from the Riemann-Roch theorem~\eqref{RR}. In the third step we have used the fact that $C$ is an exceptional curve, that is a curve with genus $g=1$ and self-intersection $C^2=-1$, which implies that $C\cdot K_{\dps{n}}=1$. The last statement follows immediately from the adjunction formula 
\begin{equation}
g=\frac{1}{2}(K_\surf\cdot C+C\cdot C)+1
\end{equation}
for the genus $g$ of a curve $C\subset S$ on a surface $\surf$. 

The above result shows that every term in the sum in Eq.~\eqref{h0dP2} can be written as an index. Is this  perhaps the case for the entire sum? A natural guess for how a single index might capture the entire expression in Eq.~\eqref{h0dP2} is
\be
\ind\bigg( D + \sum_{C \in\hat{\moricn}}\theta\left(-D \cdot C\right)(D \cdot C) C \bigg)
\stackrel{?}{=}
 \ind(D) + \frac{1}{2}\sum_{C \in\hat{\moricn}} \theta\left(-D \cdot C\right)
\left(D \cdot C\right)\left(D \cdot C + 1\right) \,.
\label{eq:poss_indrel}
\ee
Is this correct? Working out the left-hand side of Eq.~\eqref{eq:poss_indrel} by using Eq.~\eqref{eq:ind_dif}, one finds
\be
\ind\bigg( D + \sum_{C\in\hat{\moricn}_D}(D \cdot C) C \bigg)
=
\ind(D) + \frac{1}{2}\sum_{C\in\hat{\moricn}_D}(D \cdot C)\left[2(D \cdot C)+\sum_{C'\in\hat{\moricn}_D} (D \cdot C')(C \cdot C')+1 \right] \,,
\label{eq:indexp}
\ee
where we have introduced the set $\hat{\moricn}_D=\{C\in\hat{\moricn}\,|\, C\cdot D<0\}$ for ease of notation. Unfortunately, the right-hand sides of Eqs.~\eqref{eq:poss_indrel} and \eqref{eq:indexp} are not quite the same. However, if any two distinct exceptional divisors $C,C'\in\hat{\moricn}_D$ satisfy $C\cdot C'=0$ we have a perfect match. A quick look at the exceptional divisors in Appendix~\ref{app:dp_surf} shows that not all distinct pairs have a vanishing intersection. But this is also too strong a requirement - all we need is that any two distinct exceptional divisors $C$ and $C'$ with $D\cdot C<0$ and $D\cdot C'<0$ for a given effective divisor $D$ do not intersect. Remarkably, this weaker statement turns out to be true as shown in the following theorem.
\begin{thm}
Let $C$ and $C'$ be distinct generators of the Mori cone of $\dps{n}$ such that $C \cdot C' \neq 0$. If $D \cdot C \leq 0$ and $D \cdot C' < 0$ then $D$ is not in the cone of effective divisors.
\label{thm:twonegints_moricone}
\end{thm}
\begin{proof}
Assume, for contradiction, that $D$ is in the cone of effective divisors. Denote the generators of the effective cone by $C_i$ and set $C_1=C$ and $C_2=C'$ for convenience. Then we can write $D=\sum_i\alpha_i C_i$, where all $\alpha_i\geq 0$. It follows that
\[
0\geq D\cdot C_1=-\alpha_1+\alpha_2 (C_1\cdot C_2)+\sum_{i>2}\alpha_i(C_1\cdot C_i)\geq -\alpha_1+\alpha_2\, ,
\]
and, hence, that $\alpha_1\geq \alpha_2$. An analogous calculation using $0>D\cdot C_2$ leads to $\alpha_2>\alpha_1$ which is a contradiction. Hence, $D$ is not effective.
\end{proof}
\noindent Therefore, two distinct exceptional divisors $C,C'\in\hat{\moricn}_D$ do indeed satisfy $C\cdot C'=0$ and using this fact (together with $C^2=-1$) on the right-hand side of Eq.~\eqref{eq:indexp} shows that Eq.~\eqref{eq:poss_indrel} is indeed correct. This allows us to write our cohomology formula in its final form as
\be
h^0\left(\dps{n},\mc{O}_{\dps{n}}(D)\right) = \left.
  \begin{cases}
   \displaystyle{\ind\bigg( D + \sum_{C\in\hat{\moricn}}\theta\left(-D \cdot C\right)(D \cdot C) C \bigg)}
    & \textrm{if } D\in{\moricn}(\dps{n})  \\
    0
    & \textrm{otherwise}
  \end{cases}
  \right. \, .
\label{eq:dpsurf_indmastform}
\ee
We recall that $\hat{\moricn}$ is the list of Mori cone generators for $\dps{n}$, given in Appendix~\ref{app:dp_surf}. At this stage, the above formula is still a conjecture, since it is ultimately based on analysing cohomology data. Nevertheless it is quite remarkable in a number of ways. First of all, it allows for a practical computation of the zeroth cohomology dimensions even for the del Pezzo surfaces with larger Picard numbers. For a given divisor $D$ it is easy to identify the Mori cone generators $C$ with a negative intersection, $D \cdot C<0$, which enter the sum in Eq.~\eqref{eq:dpsurf_indmastform}. This is quite unlike the piecewise quadratic formulae, such as Eq.~\eqref{eq:dp2_polys}, considered earlier which become quickly unmanageable as the Picard number increases. 

Secondly, it is surprising that the zeroth cohomology dimension of an effective divisor $D$ can be written as the index,
\begin{equation}
  h^0\left(\dps{n},\mc{O}_{\dps{n}}(D)\right)={\rm ind}(\tilde{D})\, ,
\end{equation}
 of a different, shifted divisor
 \begin{equation}
  \tilde{D}= D + \sum_{C\in\hat{\moricn}}\theta\left(-D \cdot C\right)(D \cdot C) C \, . \label{shiftdP}
\end{equation} 
This is in line with what we have found for Hirzebruch surfaces in Eq.~\eqref{eq:hirzsurf_indmastform} and evidence for a more general mathematical statement is now mounting. We will derive this general mathematical result in the next section and show that it implies the above index formula for del Pezzo surfaces as well as the earlier one for Hirzebruch surfaces.


\section{Theorems for general surfaces}
\label{sec:genthms}
In the previous section, we have studied line bundle cohomology on Hirzebruch and del Pezzo surfaces, starting from algorithmically computed data. After several steps of re-writing our empirical formulae we have arrived at a remarkable result. The dimension of the zeroth cohomology can be written as an index,
\be
h^0\big(\surf,\mc{O}_\surf(D)\big) = \ind \big( \shfdiv{D} \big) \, , \label{h0ind}
\ee
of a shifted divisor $\tilde{D}$. For Hirzebruch surfaces this shifted divisor\footnote{In the companion paper \cite{mathpaper} the divisor $\tilde D$ is called the isoparametric transform of $D$.} has been defined in Eq.~\eqref{shiftFn} and the analogous result for del Pezzo surfaces is given in Eq.~\eqref{shiftdP}.

The main purpose of this section is to develop the mathematics underlying Eq.~\eqref{h0ind} in as much generality as possible and to find proofs for the Hirzebruch and del Pezzo index formulae. It turns out that the argument naturally proceeds in two steps. The first step, discussed in the following subsection, is to introduce a certain divisor shift which can be shown to leave the zeroth cohomology dimension unchanged. The second step is taken in Section~\eqref{sec:comb_with_vanthm} where we combine the divisor shift with certain vanishing theorems. As we will see, this will lead to index formulae for certain classes of surfaces, including Hirzebruch and del Pezzo surfaces. For general mathematical background see, for example, Refs.~\cite{Hartshorne1977,griffiths2014principles}.


\subsection{Cohomology-preserving shifts}
\label{sec:cohpresshif}
For both Hirzebruch and del Pezzo surfaces we have seen that a certain divisor shift plays a crucial role in writing down an index formula. Moreover, the equations~\eqref{shiftFn} and \eqref{shiftdP} for these shifts have a similar structure which suggests there exists a generalisation to all complex surfaces. The following theorem defines this general shift and asserts that it leaves the dimension of the zeroth cohomology unchanged.
\begin{thm}
Let $D$ be an effective divisor on a smooth compact complex projective surface $\surf$,  with associated line bundle $\mc{O}_\surf(D)$. Let ${\cal I}$ be the set of  irreducible negative self-intersection divisors. Then the following map on the Picard lattice,
\be
D \to \shfdiv{D} = D - \sum_{C\in{\cal I}} \theta( - D \cdot C ) \, \ceil{\frac{D \cdot C}{C^2}}C \,, \label{DDtilde}
\ee
preserves the zeroth cohomology,
\be
h^0\big(\surf,\mc{O}_\surf(\shfdiv{D})\big) = h^0\big(\surf,\mc{O}_\surf(D)\big) \,.
\ee
\label{thm:shift}
\end{thm}
\vspace{-.6cm}
\noindent While it can happen that there are infinitely many irreducible negative self-intersection divisors, only finitely many can have a negative intersection with a given divisor $D$. This means only finitely many terms appear in the sum in Eq.~\eqref{DDtilde}. Note that once the intersection form and the negative self-intersection divisors  on $\surf$ are known it is straightforward to evaluate Eq.~\eqref{DDtilde} explicitly. We will sometimes refer to Eq.~\eqref{DDtilde} as the ``master formula" for cohomology.\\[2mm]
A mathematical proof of Theorem~\ref{thm:shift} is given in an accompanying paper \cite{mathpaper} and a proof sketch, by an alternative method, is provided in Appendix~\ref{app:prfsketch}. Here, we would like to provide an intuitive explanation.

In the next two paragraphs, the term ``divisor" will refer to an actual divisor, rather than to a divisor class as in the rest of the paper. First recall that in the context of the divisor line bundle correspondence the projectivisation of the zeroth cohomology $H^0(\surf,\mc{O}_\surf(D))$ can be identified with the linear system, $|D|$, of the associated divisor. The linear system $|D|$ of a divisor consists of all effective divisors equivalent to $D$ and we can, loosely, think of it as the deformations of $D$. The dimension of the linear system and, hence, the dimension of the zeroth cohomology remains unchanged if we remove from $D$ a piece without deformations, that is, a rigid piece. 

How can a rigid piece in a divisor $D$ be detected? A rigid divisor $C$ has negative self-intersection, $C^2<0$. If such a rigid divisor $C$ is contained in $D$ it gives a negative contribution to the intersection number $D\cdot C$. Of course this negative contribution might be overwhelmed by other positive ones but if it so happens that $D\cdot C<0$ we can conclude that $D$ contains the rigid divisor $C$. This is the detection method for rigid divisors underlying Theorem~\ref{thm:shift}, as the step function in Eq.~\eqref{DDtilde} indicates. In fact, the value of the ceiling function in Eq.~\eqref{DDtilde} gives the multiple of $C$ contained in $D$. Eq.~\eqref{DDtilde} removes the multiples of all rigid divisors in $D$ which can be detected in this manner and, hence, the dimension of the zeroth cohomology remains unchanged.\\[2mm]
It is important to note that iterating the map~\eqref{DDtilde} is not necessarily trivial. The rigid pieces which can be detected in the divisor $\tilde{D}$, obtained after applying the map to $D$ once, might well be different from the ones detected in $D$. Hence, we should apply the map~\eqref{DDtilde} multiple times until the result stabilises. We denote the divisor which results from this process by $\underline{\tilde{D}}$ and this divisor has the following property.
\begin{crl}
Write $\shfdivstab{D}$ for the divisor that is the result of iterating the map $D \to \shfdiv{D}$ defined by Eq.~\eqref{DDtilde}, until stabilisation after a finite number of steps. Then $\shfdivstab{D}$ is a nef divisor such that $h^0\big(\surf,\mc{O}(D)\big) = h^0\big(\surf,\mc{O}(\shfdivstab{D})\big)$.
\label{crl:algm}
\end{crl}
\noindent It is clear from Theorem~\ref{thm:shift} {that $\underline{\tilde{D}}$ has the same zeroth cohomology dimension as $D$ but why is $\underline{\tilde{D}}$ a nef divisor? Recall that, by definition, a divisor ${\cal D}$ is  nef if there are no irreducible, negative self-intersection divisors $C$ with ${\cal D}\cdot C<0$. By construction, the divisor $\underline{\tilde{D}}$ has precisely this property.\\[2mm]
For the purpose of computing cohomology, Theorem~\ref{thm:shift} and its corollary can be helpful if the cohomology dimension of the new divisor $\underline{\tilde{D}}$ is easier to determine than that of $D$. This can happen if a suitable vanishing theorem applies to the nef divisors $\underline{\tilde{D}}$ and this is what we will discuss in the next sub-section.


\subsection{Combination with vanishing theorems}
\label{sec:comb_with_vanthm}
In the previous sub-section we have seen that  the problem of computing the zeroth cohomology dimensions over the full Picard lattice reduces to computing these cohomologies in the nef cone. Frequently, there is a vanishing theorem which assert that higher cohomologies vanish for nef divisors ${\cal D}$. In this case, the zeroth cohomology for such divisors can be computed from the index, that is, if
\begin{equation}
 h^q(\surf,\mc{O}_\surf({\cal D}))=0\mbox{ for }q=1,2\quad\Rightarrow\quad h^0(\surf,\mc{O}_\surf({\cal D}))={\rm ind}({\cal D})\, .
\end{equation} 
Hence, we have the following simple corollary.
\begin{crl}
If a vanishing theorem on a smooth compact complex projective surface $\surf$ establishes that higher cohomologies vanish in the nef cone, then any effective divisor $D$ satisfies $h^0\big(\surf,\mc{O}(D)\big) = \ind\big(\surf,\mc{O}(\shfdivstab{D})\big)$, where $\shfdivstab{D}$ is the divisor obtained from $D$ by iterating the map~\eqref{DDtilde}.
\label{crl:vanthm_givesindalg}
\end{crl}
\noindent
Which known vanishing theorems might be used to establish the required vanishing property in the nef cone? The prototypical example of a vanishing theorem for higher cohomologies is the Kodaira vanishing theorem and a particularly powerful generalisation of this theorem is the Kawamata-Viehweg vanishing theorem.
\begin{thm}[Kawamata-Viehweg vanishing theorem for surfaces]
Let $\surf$ be a smooth complex projective surface, and let $D$ be a nef and big\footnote{A nef divisor is big if and only if its self-intersection is strictly positive.} divisor on $\surf$. Then
\be
h^q\left(\surf,\mc{O}_\surf(K_\surf+D)\right) = 0 \quad \mathrm{for}\; q>0 \,.
\ee
\label{thm:kv_vanthm}
\end{thm}
\vspace{-.8cm}
\noindent When a space is toric, there is the stronger Demazure vanishing theorem.
\begin{thm}[Demazure vanishing theorem for surfaces] \label{demazure}
Let $\surf$ be a toric surface whose fan has convex support, and let $D$ be a nef divisor. Then
\be
h^q\left(\surf,\mc{O}_\surf(D)\right) = 0 \quad \mathrm{for}\; q>0 \,.
\ee
\label{thm:d_vanthm}
\end{thm}
\vspace{-.6cm}
\noindent
For which spaces do these theorems guarantee that higher cohomologies vanish in the nef cone? The Demazure vanishing theorem holds for the entire nef cone and applies to toric varieties whose fans have convex support. In fact, the fan has convex support for any compact toric surface - see for example Ref.~\cite{cox2011toric} - so we have the following corollary.
\begin{crl}
Let $\surf$ be a compact toric surface, and let $D$ be an effective divisor. Then
\be
h^0\big(\surf,\mc{O}(D)\big) = \ind(\shfdivstab{D}) \, ,
\ee
where the divisor $\shfdivstab{D}$ is obtained from $D$ by iterating the map~\eqref{DDtilde}.
\label{crl:torspac_ind}
\end{crl}
\noindent
Clearly, this covers a large and important set of surfaces, including all Hirzebruch surfaces, as well as their blow-ups, the del Pezzo surfaces $\dps{n}$ for $n=1,2,3$,  and their blow-ups and the toric surfaces that correspond to the 16 reflexive polytopes. What we have shown is that index formulae for the zeroth cohomology dimension exist for all these cases. Since all these toric surfaces frequently appear in compactification these results are of direct relevance for string theory.\\[2mm]
The Kawamata-Viehweg vanishing theorem applies to a very general class of surfaces $\surf$, but it guarantees vanishing in a region that is not precisely the nef cone. Specifically, it asserts vanishing of the higher cohomologies of a divisor $D$ if $D - K_\surf$ is nef and big. It applies to divisors in the intersection of the nef and big cones, shifted by the anti-canonical divisor $-K_\surf$. This is of partial use, since this region has some overlap with the nef cone and, hence, leads to index formulae for some but not all effective divisors. Sometimes this overlap covers almost all of the nef cone and this can be shown to happen for Hirzebruch surfaces.

However, for surfaces with a nef and big anti-canonical bundle $-K_\surf$ the Kawamata-Viehweg vanishing theorem can be applied to all nef divisors $D$. To see this, first note that both $D$ and $-K_\surf$ being nef immediately impies that $D-K_\surf$ is nef.  To show that $D-K_\surf$ is big we consider its self-intersection
\be
(D - K_\surf)^2 = D^2 + 2 D \cdot (-K_\surf) + (- K_\surf)^2 \, .
\ee
Since $-K_\surf$ is nef and $D$ is effective, the first two terms on the right-hand-side are $\geq 0$. Additionally since $-K_\surf$ is big, $(-K_\surf)^2 > 0$, so that $(D - K_\surf)^2>0$ and, hence, $D-K_\surf$ is big. This leads to the following corollary.
\begin{crl} \label{crldPvan}
Let $\surf$ be a smooth compact complex projective surface with nef and big anti-canonical divisor, $-K_\surf$, and let $D$ be an effective divisor. Then
\be
h^0\big(\surf,\mc{O}(D)\big) = \ind(\shfdivstab{D}) \, ,
\ee
where the divisor $\shfdivstab{D}$ is obtained from $D$ by iterating the map~\eqref{DDtilde}.
\end{crl}
\noindent 
A quick glance at Eq.~\eqref{KdP} shows that the anti-canonical divisor for $\dps{n}$ is nef and big. Hence Corollary~\ref{crldPvan} applies to all del Pezzo surfaces $\dps{n}$, where $n=0,1,\ldots ,8$, and guarantees the existence of an index formula.\\[2mm]
We have now seen that index formulae exist for many types of surfaces. However, the divisor in the argument of the index is obtained by iterating the master formula~\eqref{DDtilde} and is, hence, fairly complicated in general.  In practice, this corresponds to an algorithm for computing the cohomology dimension rather than explicit formula. However, there are surfaces where a single application of the master formula~\eqref{DDtilde} already maps into the nef cone. In such cases, we have a simple index formula for the cohomology dimension as stated in the following corollary.
\begin{crl}\label{cor:single}
Let $\surf$ be a smooth compact complex projective surface on which a vanishing theorem guarantees that higher cohomologies vanish in the nef cone. If additionally, for any effective divisor $D$ the divisor $\tilde{D}$ defined in Eq.~\eqref{DDtilde} is in the nef cone, then we have
\be\label{crl:oneshift_indform}
h^0\left(\surf,\mc{O}_\surf(D)\right) = \ind\bigg( D - \sum_{C\in{\cal I}}\theta( - D \cdot C ) \, \ceil{\frac{D \cdot C}{C^2}}C \bigg) \,,
\ee
where ${\cal I}$ is the set of irreducible negative self-intersection divisors.
\end{crl}
\noindent
We have already established for both Hirzebruch and del Pezzo surfaces $\dps{n}$ that higher cohomologies vanish in the nef cone. Now we show that the second condition in Corollary~\ref{cor:single} is also satisfied, that is, the master formula~\eqref{DDtilde} maps into the nef cone for both classes of surfaces.

This is quite easy to see for Hirzebruch surfaces $\mbb{F}_n$ since there exists only one irreducible negative self-intersection divisor $C=D_1$ with $C^2=-n$. The master formula removes as many copies of $C$ from $D$ as are detected by the intersection number $D\cdot C$. So, even if $D\cdot C<0$ we have
\be
\shfdiv{D} \cdot C = D \cdot C - \ceil{\frac{ D \cdot C}{C^2}}C^2 \geq 0 \,.
\ee
This means applying the master formula once more to $\tilde{D}$ leaves the divisor unchanged. Hence, from Corollary~\ref{crl:oneshift_indform}, we have the following index formula for the Hirzebruch surfaces.
\begin{crl}
Let $D$ be any effective divisor on the Hirzebruch surface $\mbb{F}_n$. Then
\be
h^0\left(\mbb{F}_n,\mc{O}_{\mbb{F}_n}(D)\right) = \ind\bigg( D - \theta(-D\cdot C) \, \ceil{\frac{D \cdot C}{C^2}}C \bigg) \,,
\ee
where $C$ is the unique irreducible negative self-intersection divisor, for which $C^2 = -n$.
\label{crl:hirz_oneshift_indform}
\end{crl}
\noindent
This proves the index formula for Hirzebruch surfaces, conjectured from cohomology data in Eq.~\eqref{eq:hirzsurf_indmastform}.\\[2mm]
For del Pezzo surfaces $\dps{n}$ we can also show that the map~\eqref{DDtilde} stabilises after the first step and, hence, that the divisor $\tilde{D}$ is in the nef cone.
However, the proof is slightly more difficult than for Hirzebruch surfaces due to the presence of multiple irreducible divisors with negative self-intersection.  We refer to the accompanying paper \cite{mathpaper} for a proof and simply cite the following result.
\begin{thm}
On a smooth compact complex projective surface containing no irreducible divisors with self-intersection $<-1$, the divisor $\shfdiv{D}$ defined in Eq.~\eqref{DDtilde} is nef for all effective divisors $D$.
\end{thm}
\noindent
It is a well-known statement that the irreducible negative self-intersections divisors $C$ of $\dps{n}$ satisfy $C^2=-1$, so the above theorem applies to del Pezzo surfaces. Together with Corollary~\ref{crl:oneshift_indform} this leads to the following index formula for del Pezzo surfaces $\dps{n}$.
\begin{crl}
Let $\dps{n}$ be the del Pezzo surface constructed by blowing up $n$ generic points of the complex projective plane $\mbb{P}^2$, where $n\in\{0,1,\ldots ,8\}$, and let $D$ be any effective divisor. Then
\be
h^0\left(\dps{n},\mc{O}_{\dps{n}}(D)\right) = \ind\bigg( D + \sum_{C\in{\cal I}} \theta( - D \cdot C ) \, (D \cdot C)C \bigg) \,,
\ee
where ${\cal I}$ is the set of irreducible negative self-intersection divisors.
\label{crl:dpn_oneshift_indform}
\end{crl}
\noindent This proves the result we have conjectured in Eq.~\eqref{eq:dpsurf_indmastform}. Note that the ceiling function present in the master formula~\eqref{DDtilde} disappears for del Pezzo surfaces since $C^2=-1$ for all irreducible, negative self-intersection curves $C$.


\section{Example applications}\label{sec:app}
In Section~\ref{sec:dphirz_mastform} we have described the path from cohomology data to index formulae in some detail and in the previous section we have provided proofs for these index formulae. Here, we would like to take the mathematical results of the previous section for granted and illustrate how to derive cohomology formulae for specific surfaces from these general results. The examples we discuss are the Hirzebruch surface $\mbb{F}_2$ and the del Pezzo surface $\dps{2}$.


\subsection{The Hirzebruch surface $\mbb{F}_2$}
Much of the relevant information for Hirzebruch surfaces has already been summarised in Section~\ref{sec:Hirze}. However, all Hirzebruch surfaces $\mbb{F}_n$ are toric 
and therefore provide a good illustration for how to extract the relevant information from toric data. We will briefly discuss how to do this, based on the information provided in Appendix~\ref{app:hirz_surf}.

First recall the following basic facts from toric geometry. The Picard lattice of a toric surface is spanned by integer combinations of the toric divisors, which correspond to rays of the toric diagram. Equivalence relations between toric divisors follow from the weight system and choosing representatives for the equivalence classes leads to a basis for the Picard lattice. For $\mathbb{F}_2$ we choose a basis $(D_1,D_2)$ of the Picard lattice as explained in Appendix~\ref{app:hirz_surf} and in line with our conventions in Section~\ref{sec:Hirze}. A general divisor is written as $D=k_1D_1+k_2D_2$, so is labelled by a two-dimensional integer vector ${\bf k}=(k_1,k_2)$, and the Mori cone, ${\moricn}(\mbb{F}_2)$, consists of all divisors with $k_1\geq 0$ and $k_2\geq 0$. The intersections between toric divisors are easily extracted: distinct toric divisors intersect if they correspond to neighbouring rays in the diagram, and self-intersections are determined by equivalence relations between toric divisors. From Appendix~\ref{app:hirz_surf}, this leads to the $\mbb{F}_2$ intersection rules
\begin{equation}
 D_1^2=-2\,,\quad D_1\cdot D_2=1\,,\quad D_2^2=0\, ,
\end{equation}
in line with Eq.~\eqref{Hirzeisec}.  From Eq.~\eqref{nefFn}, the generators of the nef cone ${\nefcn}(\mbb{F}_2)$ are given by 
\begin{equation}
\hat{\nefcn}(\mbb{F}_2)=\{D_2,D_1+2D_2\}\, .
\end{equation}
It is straightforward to verify that the irreducible negative self-intersection divisors are precisely the toric divisors with negative self-intersection. For Hirzebruch surfaces there is precisely one such divisor, namely $C=D_1$. This is all the information we need to evaluate the index formula~\eqref{crl:oneshift_indform}.\\[2mm]
We refer to the region of non-effective divisors as Region 3, so
\begin{equation}
 D\mbox{ not effective}\quad\Leftrightarrow\quad k_1<0\mbox{ or }k_2<0\quad \mbox{(Region 3)}\, .
\end{equation} 
Obviously, in this region we have $h^0(\mbb{F}_2,\mc{O}_{\mbb{F}_2}({\bf k}))=0$. The Mori cone splits into precisely two regions, since there is only one irreducible negative self-intersection divisor. These two regions are given by
\begin{eqnarray} 
 D\mbox{ effective and }D\cdot C\geq 0&\Leftrightarrow&k_1\geq 0,\,k_2\geq 0,\,-2k_1+k_2\geq 0\quad\mbox{(Region 1)}\\
 D\mbox{ effective and } D\cdot C< 0&\Leftrightarrow&k_1\geq 0,\, k_2\geq 0,\,-2k_1+k_2<0\quad\mbox{(Region 2)}\, .
\end{eqnarray} 
Region 1 is the nef cone. The master formula~\eqref{DDtilde} in this region is simply the identity map and, hence, the zeroth cohomology dimension is given by the index. In Region 2, the master formula specialises to
\be
\shfdiv{D} = D - \ceil{\frac{D \cdot C}{C^2}}C = \left( k_1- \ceil{\frac{nk_1-k_2}{n}}\right)D_1 + k_2 D_2 \,. \label{F2shift}
\ee
and the zeroth cohomology of $D$ is given by the index of this shifted divisor $\tilde{D}$. As illustrated in Fig.~\ref{fig:f2_piclat}, Eq.~\eqref{F2shift} corresponds to a shift from Region 2 into the nef cone (Region 1).
\begin{figure}[h]
  \includegraphics[scale=.45]{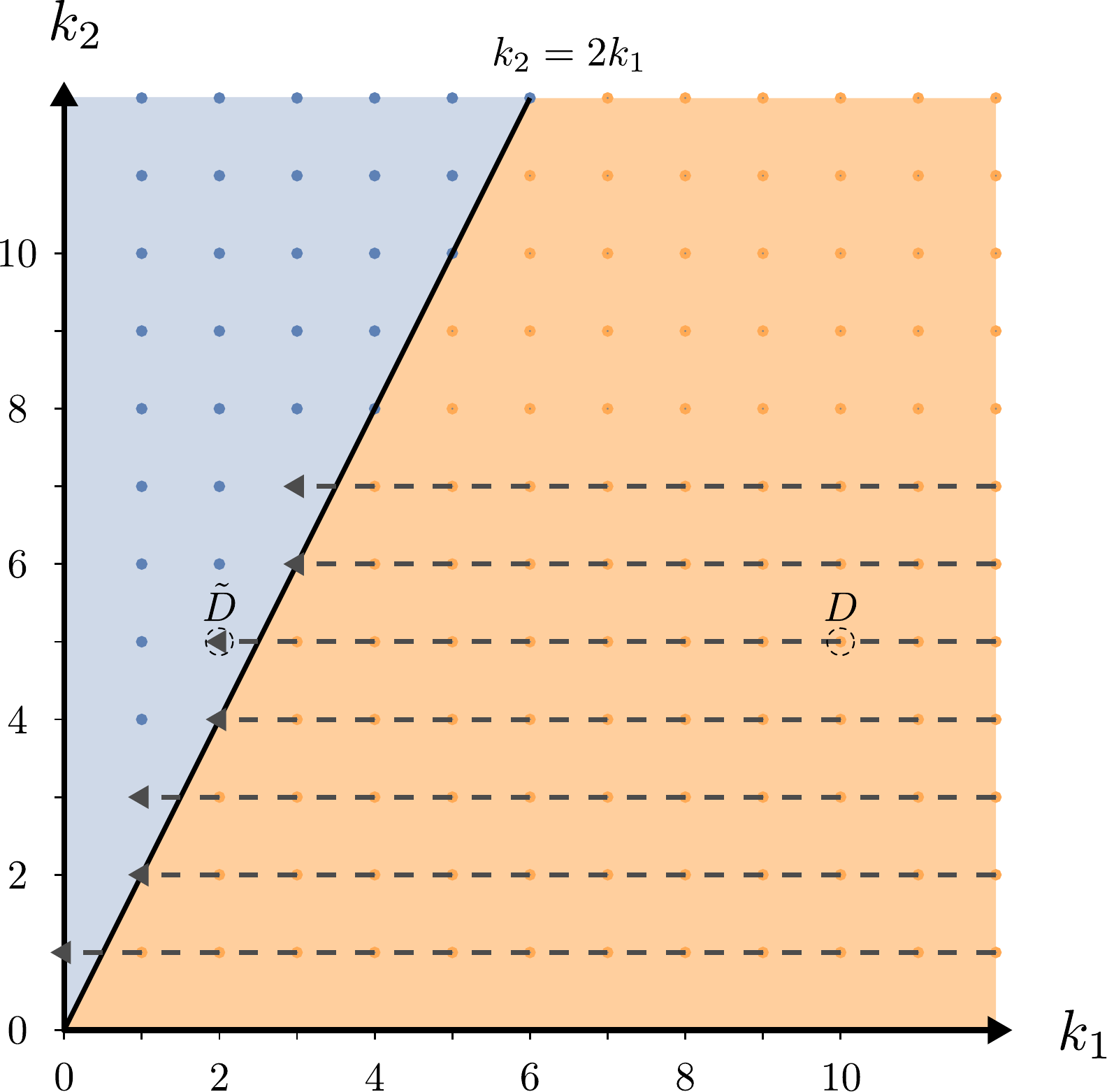}
  \caption{\sf The Picard lattice of the Hirzebruch surface $\mbb{F}_2$. The positive quadrant is the Mori cone - outside of it the zeroth cohomology vanishes. The nef cone is the blue area. We show the shifts $D \to \shfdiv{D}$ that map into the nef cone. Note every point on a dotted line maps to the same final divisor.}
  \label{fig:f2_piclat}
\end{figure}
Combining this information with the explicit expression~\eqref{indFn} for the index then leads to the formula
\be
\begin{aligned}
h^0\left(\mbb{F}_n,\mc{O}_{\mathbb{F}_n}({\bf k})\right) &= \left.
  \begin{cases}
   \displaystyle{{\rm ind}(\mc{O}_{\mbb{F}_2}({\bf k}))}
    & \textrm{(Region 1)}\\
   \displaystyle{{\rm ind}(\mc{O}_{\mbb{F}_2}({\bf k}))+\frac{1}{2}nc-\frac{1}{2}nc^2-ck_2+nck_1-c}
    & \textrm{(Region 2)}\\
    0& \textrm{(Region 3)}
   \end{cases}
  \right. \\
c&= \ceil{\frac{nk_1-k_2}{n}}\, .
\end{aligned}
\ee
This formula is consistent with the earlier one, Eq.~\eqref{h0Hirze}, which we have extracted from cohomology data.


\subsection{The del Pezzo surface $\dps{2}$}

On del Pezzo surfaces $\dps{n}$, the Mori cone, irreducible negative self-intersection divisors, and the intersection form are well-known and we have collected these in Section~\ref{sec:dPprop} and in Appendix~\ref{app:dp_surf}. Let us recall what these look like for the case of $\dps{2}$. The rank of the $\dps{2}$ Picard lattice is three and a basis of divisors is given by
\be
D_0 = \hcl \,, \quad D_1 = \ecl_1 \,, \quad D_2 = \ecl_2 \,,
\ee
where $\hcl$ is the hyperplane class of the underlying complex projective plane $\mbb{P}^2$, and $\ecl_1$ and $\ecl_2$ are the exceptional classes of the two blow-ups. The intersection form is fixed by the relations
\be
l^2=1\,,\quad l\cdot \ecl_i=0\,,\quad \ecl_i\cdot \ecl_j=-\delta_{ij}\, .
\ee
A general divisor is written as $D=k_0D_0+k_1D_1+k_2D_2$ and is, hence, labelled by a three-dimensional integer vector ${\bf k}=(k_0,k_1,k_2)$. The Mori and nef cone generators are
\begin{eqnarray}
 \hat{\moricn}(\dps{2})&=&\{C_1=\ecl_1,\,C_2=\ecl_2,\,C_3=l-\ecl_1-\ecl_2\}\\
 \hat{\nefcn}(\dps{2})&=&\{N_1=l-\ecl_1,\,N_2=l-\ecl_2,\,N_3=l\}
\end{eqnarray} 
The irreducible negative self-intersection divisors are the exceptional curves which are precisely the Mori cone generators, $C\in\hat{\moricn}(\dps{2})$. More explicitly, the Mori cone is characterised by the equations
\begin{equation}
\left.\begin{array}{l}D\cdot N_1\geq 0\\D\cdot N_2\geq 0\\D\cdot N_3\geq 0\end{array}\right\}\quad\Leftrightarrow\quad
\left\{\begin{array}{l}k_0+k_1\geq 0\\k_0+k_2\geq 0\\k_0\geq 0\end{array}\right.
\end{equation}
and line bundles $\mc{O}_{\dps{2}}({\bf k})$ outside this cone have vanishing zeroth cohomology.\\[2mm]
As for Hirzebruch surfaces, the master formula~\eqref{DDtilde} maps into the nef cone, without the need for iteration, and, from Corollary~\ref{crl:dpn_oneshift_indform}, the resulting index formula for effective divisors $D$ is
\be
h^0\left(\dps{n},\mc{O}_{\dps{n}}(D)\right) = \ind\bigg( D + \sum_{i=1}^3 \theta( - D \cdot C_i ) \, (D \cdot C_i)C_i \bigg) \, .
\label{eq:dpn_indform}
\ee
The regions within the Mori cone are separated off from one another by hyperplanes which can be read off from the arguments of the three theta functions in Eq.~\eqref{eq:dpn_indform}. These hyperplanes are
\begin{equation}\left.
\begin{array}{l}
 D\cdot C_1=0\\D\cdot C_2=0\\D\cdot C_3=0
\end{array}\right\}\quad\Leftrightarrow\quad
\left\{
\begin{array}{l}
k_1=0\\k_2=0\\k_0+k_1+k_2=0
\end{array} \right.
\end{equation}
In principle, these three hyperplanes define eight regions but only the following five intersect the Mori cone.
\begin{equation}\label{dP2regions2}
(D\cdot C_1,D\cdot C_2,D\cdot C_3)\,\left\{
\begin{array}{l}
 (\geq,\geq,\geq)\\
 (\geq,\geq,<)\\
 (<,\geq,\geq)\\
 (\geq,<,\geq)\\
 (<,<,\geq)
\end{array} \right\}\, 0\quad
\begin{array}{l}
\mbox{(Region 1)}\\
\mbox{(Region 2)}\\
\mbox{(Region 3)}\\
\mbox{(Region 4)}\\
\mbox{(Region 5)}
\end{array}
\end{equation} 
Inserting the $C_i$ into the master formula~\eqref{DDtilde}, we find the explicit expression for the divisor shifts from those five regions into the nef cone (Region 1).
\begin{equation}\label{dP2shifts}
 \tilde{D}=\left\{
 \begin{array}{l}
  D\\D+(D\cdot C_3)C_3\\D+(D\cdot C_1)C_1\\D+(D\cdot C_2)C_2\\D+\sum_{i=1}^2(D\cdot C_i)C_i
 \end{array}\right\}\,\leftrightarrow\,\left\{
\begin{array}{l}
 (k_0,k_1,k_2)\\
 (2k_0+k_1+k_2,-k_0-k_2,-k_0-k_1)\\
 (k_0,0,k_2)\\
 (k_0,k_1,0)\\
 (k_0,0,0)
\end{array} \right\}
\begin{array}{l}
\mbox{(Region 1)}\\
\mbox{(Region 2)}\\
\mbox{(Region 3)}\\
\mbox{(Region 4)}\\
\mbox{(Region 5)}
\end{array}
\end{equation} 
Here, the second bracket lists the coordinate vectors of the shifted divisors $\tilde{D}$ relative to the basis $(l,\ecl_1,\ecl_2)$.  The five regions together with the coordinate shifts are indicated in Fig.~\ref{dp2_piclat_numb_shifts}. The number of terms required for each shift equals the number of hyperplanes which have to be crossed to reach the nef cone.  Also note that the shifts map precisely to the boundaries of the nef cone.
\begin{figure}[!h]
  \includegraphics[scale=0.7]{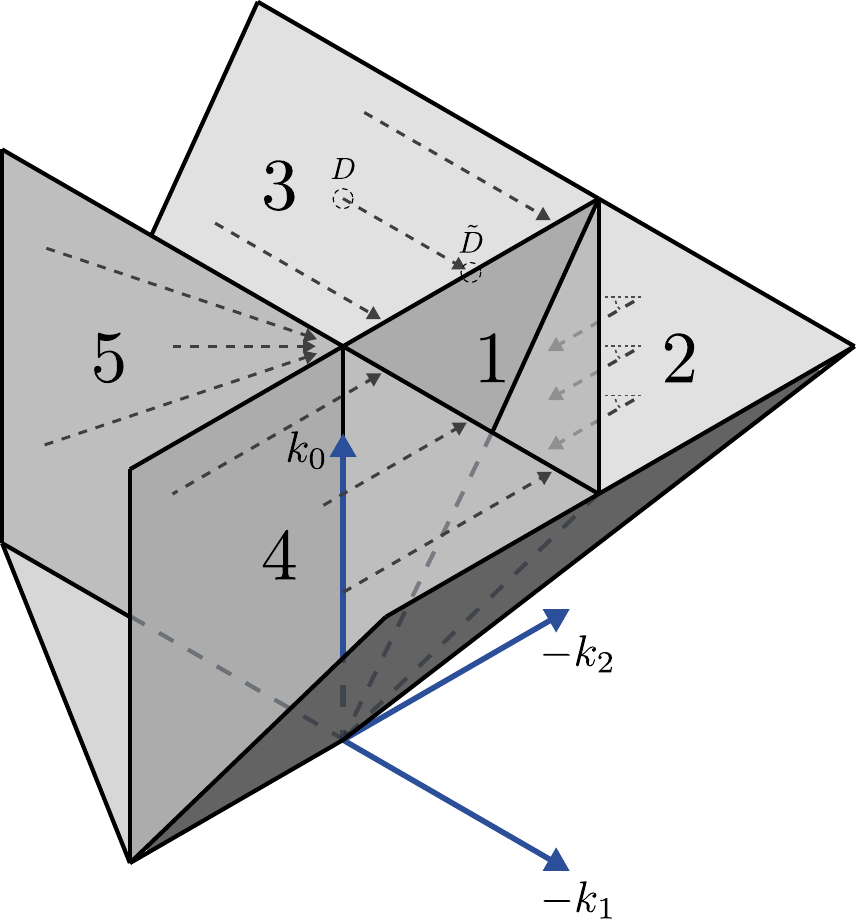}
  \caption{\sf The Picard lattice of the del Pezzo surface $\dps{2}$. We do not draw the lattice points to avoid clutter. The five numbered cones are the regions defined in Eq.~\eqref{dP2regions2} which, together, make up the Mori cone. Region 1 is the nef cone. In each region a different quadratic polynomial describes the zeroth cohomology.  The arrows indicate the divisor shifts~\eqref{dP2shifts} into the nef cone.}
  \label{dp2_piclat_numb_shifts}
\end{figure}
Inserting these divisor shifts into the expression~\eqref{inddP} for the index finally gives the desired piecewise quadratic formula for the zeroth cohomology dimension. \\
\be
h^0\left(\dps{2},\mc{O}_{\dps{2}}({\bf k})\right) = \left.
  \begin{cases}
    \displaystyle{1+\frac{3}{2}k_0+\frac{1}{2}k_0^2+\frac{1}{2}k_1-\frac{1}{2}k_1^2+\frac{1}{2}k_2-\frac{1}{2}k_2^2}
    &\mbox{(Region 1)}  \vspace{1mm}\\
    \displaystyle{1+2k_0+k_0^2+k_1+k_0k_1+k_2+k_0k_2+k_1k_2 \phantom{\displaystyle{\frac{1}{2}}}}
    &\mbox{(Region 2)}  \vspace{1mm}\\
    \displaystyle{1+\frac{3}{2}k_0+\frac{1}{2}k_0^2+\frac{1}{2}k_2-\frac{1}{2}k_2^2}
    &\mbox{(Region 3)}  \vspace{1mm}\\
    \displaystyle{1+\frac{3}{2}k_0+\frac{1}{2}k_0^2+\frac{1}{2}k_1-\frac{1}{2}k_1^2}
    &\mbox{(Region 4)}  \vspace{1mm}\\
    \displaystyle{1+\frac{3}{2}k_0+\frac{1}{2}k_0^2}
    &\mbox{(Region 5)}
  \end{cases}
  \right.\, .
\ee
This is consistent with our earlier empirical result~\eqref{eq:dp2_polys} which was extracted from cohomology data.


\section{Conclusion}\label{sec:con}
In this paper, we have analysed line bundle cohomology on complex surfaces, focusing on the examples of Hirzebruch and del Pezzo surfaces. Starting from algorithmically computed cohomology data, we have shown how to conjecture piecewise quadratic formulae for the zeroth cohomology dimension on these manifolds. Through a process of gradual re-writing, we have brought these equations into a suggestive form where the cohomology dimension is expressed in terms of an index.

Based on these examples, we have conjectured and proved~\cite{mathpaper} that a certain shift $D\rightarrow \tilde{D}$ of effective divisors, explicitly defined in Eq.~\eqref{DDtilde}, leaves the zeroth cohomology dimension unchanged, so that $h^0(\surf,\mc{O}_\surf(D))=h^0(\surf,\mc{O}_\surf(\tilde{D}))$. Repeated application of this shift ends up with a divisor $\underline{\tilde{D}}$ in the nef cone and provided a suitable vanishing theorem exists, the zeroth cohomology dimension can, remarkably, be written in terms of the index as
\begin{equation}
 h^0(\surf,\mc{O}_\surf(D))={\rm ind}(\mc{O}_\surf(\underline{\tilde D}))\, .
\end{equation}
It turns out that a suitable vanishing theorem exists for many classes of surfaces, including Hirzebruch and del Pezzo surfaces as well as compact toric surfaces. Moreover, for Hirzebruch and del Pezzo surfaces, the nef cone is already reached after a single shift, as given in Eqs.~\eqref{shiftFn} and \eqref{shiftdP}, and this explains the explicit index formulae for these manifolds.
\vspace{8pt}

Clearly, there are many further questions suggested by our results. In the context of complex surfaces it is desirable to better understand the interplay between the divisor shift and the required vanishing theorem. Specifically, can we identify  further classes of surfaces for which index formulae can be derived in this way? The arguments in Appendix~\ref{app:prfsketch} show that more general shifts than Eq.~\eqref{DDtilde} are possible. Could such modifications of the shift formula be relevant for other classes of surfaces? For surfaces, we know that higher cohomology dimensions can be expressed in terms of the zeroth cohomology, using Serre duality and the index theorem. But it is not clear if higher cohomologies can also be written in terms of a single index formula, similar to Eq.~\eqref{crl:oneshift_indform}.

We have already noted the quasi-topological nature of the cohomology formulae: the quantity ${\rm ind}(\mc{O}_\surf(\underline{\tilde D}$)) is purely topological, while the map $D\rightarrow \tilde{D}$ depends in general on the complex structure. The complex structure dependence of the map $D\rightarrow \tilde{D}$ may greatly facilitate the study of jumping loci for cohomology (see Ref.~\cite{Anderson:2011ty} for a discussion of jumping loci in complex structure moduli stabilisation; in string model building jumping loci frequently enter the construction of Higgs fields). 

Looking further ahead, what is the situation for three-folds or complex manifolds of even higher dimension? Complex surfaces are special in that curves and divisors have the same dimension, so the results of this paper cannot be straightforwardly extended to higher-dimensional cases. On the other hand, Refs.~\cite{Constantin:2018hvl,Larfors:2019sie} show that piecewise cubic formulae for line bundle cohomology dimensions can be written down for Calabi-Yau three-folds. This strongly suggests an underlying mathematical structure, although its nature and formulation is not known at present. These questions are currently under investigation.

\section*{Acknowledgments}
We thank Fabian Ruehle for useful discussions. C.R.B. is supported by an STFC studentship and R.D.~by a Skynner Fellowship of Balliol College, Oxford.
%


%
%

\appendix


\section{Hirzebruch surfaces}
\label{app:hirz_surf}

The Hirzebruch surfaces are complex manifolds which correspond to the various ways of fibering $\mbb{P}^1$ over another $\mbb{P}^1$ base. They are denoted by $\mbb{F}_n$ and are indexed by an integer $n\geq0$ that corresponds to the number of times the fibre $\mbb{P}^1$ `twists' as one moves along the base\footnote{Up to real diffeomorphisms there are only two ways to fibre a sphere over a sphere. However up to complex diffeomorphisms there are infinitely many. The even Hirzebruch surfaces, $\mbb{F}_{2n}$, are isomorphic as real manifolds to the trivial fibre bundle, while the odd cases, $\mbb{F}_{2n+1}$, are isomorphic as real manifolds to the unique non-trivial fibre bundle.}. The Hirzebruch surfaces with small $n$ are isomorphic to other well-known spaces.  Specifically, we have $\mbb{F}_0 \iso \mbb{P}^1 \times \mbb{P}^1$ and $\mbb{F}_1 \iso \dps{1}$ where $\dps{1}$ is the del Pezzo surface obtained by blowing up $\mbb{P}^2$ at a single point. Hirzebruch surfaces can be constructed using toric geometry or complete intersections in products of projective spaces and we discuss these two possibilities in turn. 

\smlhdg{Toric representations}\\
\noindent
The toric diagram for $\mbb{F}_n$ together with its weight system is shown in Figure~\ref{fig:fn_tordiag}. The weight system, that is, the set of charges of the toric coordinates $z_1,\ldots,z_4$ under the two toric scalings, follows from linear equivalences of the toric rays.

\begin{figure}[h]
\begin{minipage}{0.30\linewidth}
\begin{center}
\includegraphics[scale=1.5]{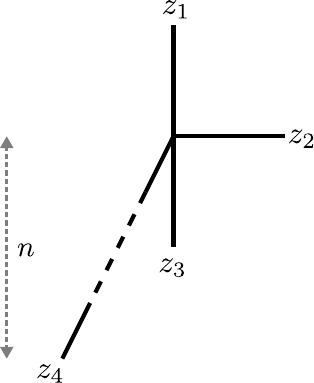}
\end{center}
\end{minipage}
\begin{minipage}{0.30\linewidth}
\begin{center}
Weight system: \vspace{.2cm} \\
\begin{tabular}{c c c c}
$\tc_1$ & $\tc_2$ & $\tc_3$ & $\tc_4$ \\ \hline
1 & 0 & 1 & 0 \\
n & 1 & 0 & 1
\end{tabular} \,,
\end{center}
\end{minipage}
\caption{Toric diagram and derived weight system for the Hirzebruch surface $\mbb{F}_n$. All distances in the toric ray diagram are 1 except the vertical value of $z_4$ which is $-n$.}
\label{fig:fn_tordiag}
\end{figure}

The weight system implies the following divisor equivalences,
\be
[\tc_2] = [\tc_4] \,, \quad [\tc_1] = [\tc_3] + n[\tc_4] \,.
\ee
As a basis of the Picard lattice, we can choose the two divisors $D_1 = [\tc_3]$ and $D_2 = [\tc_4]$. The anti-canonical divisor $-K_{\mbb{F}_n}$ is simply the sum of all toric divisors, and so in this basis it is given by
\be
-K_{\mbb{F}_n} = 2D_1 +(n+2)D_2 \,.
\ee
The intersection form can also be read off from the toric diagram: distinct toric divisors intersect once if they correspond to neighbouring rays and do not intersect otherwise, and self-intersections follow from divisor equivalences. The leads to the intersections
\be
 D_1^2=-n\,,\quad D_1\cdot D_2=1\,,\quad D_2^2=0\, .
\ee
The Mori cone is given by the non-negative linear combinations of the toric divisors and, hence, the Mori cone generators are $\hat{\moricn}(\mbb{F}_n)=\{D_1,D_2\}$. From this, the generators of the dual nef cone can be determined as $\hat{\nefcn}(\mbb{F}_n)=\{D_2,D_1+nD_2\}$.

\smlhdg{Complete intersection representations}\\
\noindent
The Hirzebruch surfaces also have complete intersection representations. A complete intersection can be specified by a configuration matrix, where the entries in the first column are the projective spaces whose product forms the ambient space, and each further column specifies the multi-degree of one of the defining equations. The Hirzebruch surface $\mbb{F}_n$ can be represented as a hypersurface in $\mbb{P}^1 \times \mbb{P}^2$, defined by the zero locus of a polynomial with bi-degree $(n,1)$. Hence, the corresponding configuration matrix is
\be
\mbb{F}_n \in
\left[
\begin{array}{c | c}
\mathbb{P}^1 & n \\
\mathbb{P}^2 & 1
\end{array}
\right] \,. \quad
\ee


\section{Del Pezzo surfaces}
\label{app:dp_surf}

Del Pezzo surfaces are defined to be smooth compact complex projective surfaces with ample anti-canonical bundle; these are the two-dimensional Fano varieties. A del Pezzo surface is either isomorphic to the product of complex projective lines $\mbb{P}^1 \times \mbb{P}^1$ or it is isomorphic to the blow-up of the complex projective plane $\mbb{P}^2$ at up to eight points in general position. Conversely, any blow-up of $\mbb{P}^2$ at eight points in general position is a del Pezzo surface. The condition that the points are in general position is equivalent to ampleness of the anti-canonical divisor class of the resulting space \cite[Th\'{e}or\`{e}me~1]{Demazure1976}.\\[2mm]
The subset of the del Pezzo surfaces that are of primary interest in this paper are those constructed by blowing up the complex projective plane $\mbb{P}^2$ at $n$ points in general position, where $n=0,\ldots ,8$. We denote these surfaces by $\dps{n}$ throughout\footnote{Sometimes the del Pezzo surfaces are numbered by their degree. The degree $\surfdegr$ of a del Pezzo surface, which is defined to be the self-intersection number of the anti-canonical divisor class, is at least 1 and at most 9. For a del Pezzo surface isomorphic to the blow-up of $\mbb{P}^2$ at $n$ points in general position, the degree is $\surfdegr = 9 - n$. When the del Pezzo surface is isomorphic to $\mbb{P}^1 \times \mbb{P}^1$ the degree is $\surfdegr=8$.}. Below we will collect relevant information about the  surfaces $\dps{n}$; note that the analogous properties for $\mbb{P}^1 \times \mbb{P}^1$ would be trivial. In particular, after recalling some basic methods to construct del Pezzo surfaces, we will list the generators of the Mori and nef cones, which we use in the main text.

\smlhdg{Picard lattice and anti-canonical divisor}\\
\noindent
There is a natural choice for a basis of the $\dps{n}$ Picard lattice which consists of the hyperplane class $\hcl$ of the underlying projective plane $\mbb{P}^2$ and the exceptional divisor class $\ecl_i$, where $i=1,\ldots ,n$, which correspond to the blow-ups. Relative to the basis $\hcl,\ecl_1,\ldots ,\ecl_n)$, the intersection form is fixed by the relations
\begin{equation}
 l^2=1\,,\quad l\cdot \ecl_i=0\,,\quad \ecl_i\cdot \ecl_j=-\delta_{ij}\, .
\end{equation} 
The anti-canonical divisor is given by
\be
-K_{\dps{n}} = 3\hcl - \sum_{i=1}^n \ecl_i \, .
\ee

\smlhdg{Toric representations}\\
\noindent
There are toric representations for a subset of the del Pezzo surfaces, namely for the spaces isomorphic to $\mbb{P}^1 \times \mbb{P}^1$ and the spaces isomorphic to $\dps{n}$, where $n=0,\ldots ,3$. In Figure~\ref{fig:dp_toricdiagsandweights} we recall the toric ray diagrams and the associated weight systems for the three non-trivial cases $\dps{1}$, $\dps{2}$, and $\dps{3}$. The weight system, that is, the set of charges of the toric coordinates under the toric scalings, follows from linear equivalences of the toric rays. 
\begin{figure}[h]
  \includegraphics[scale=1.5]{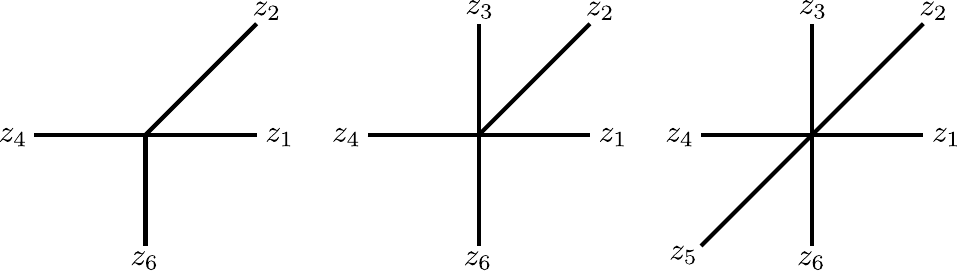}
\begin{minipage}{0.30\linewidth}
\begin{center}
\phantom{1}
Toric $\dps{1}$ diagram, \\
with weight system:
\begin{tabular} { c c c c c c }
$z_1$ 	& $z_2$ 	& \phantom{$z_3$}      	& $z_4$ 	& \phantom{$z_5$}       	& $z_6$ 	\\ \hline
0 		& 1 		&   					& 1 		& 					& 1		\\
1		& 0 		&   					& 1 		& 					& 0		\\
		&		&					&		&					&		\\
		&		&					&		&					&		\\
\end{tabular}
\end{center}
\end{minipage}
\begin{minipage}{0.30\linewidth}
\begin{center}
\phantom{1}
Toric $\dps{2}$ diagram, \\
with weight system:
\begin{tabular} { c c c c c c }
$z_1$ 	& $z_2$ 	& $z_3$   	& $z_4$ 	&  \phantom{$z_5$}     	& $z_6$ 	\\ \hline
0 		& 1 		&  0 		& 1 		& 					& 1		\\
1		& 0 		&  0		& 1 		& 					& 0		\\
0		& 0 		&  1		& 0 		& 					& 1		\\
		&		&		&		&					&		\\
\end{tabular}
\end{center}
\end{minipage}
\begin{minipage}{0.30\linewidth}
\begin{center}
\phantom{1}
Toric $\dps{3}$ diagram, \\
with weight system:
\begin{tabular} { c c c c c c }
$z_1$ 	& $z_2$ 	& $z_3$   	& $z_4$ 	& $z_5$   	& $z_6$ 	\\ \hline
0 		& 1 		& 0 		& 1 		& 0		& 1		\\
1		& 0 		& 0		& 1 		& 0		& 0		\\
0		& 0 		& 1		& 0 		& 0		& 1		\\
0		& 1		& 0		& 0		& 1		& 0		\\
\end{tabular}
\end{center}
\end{minipage}
\caption{Toric diagrams and weight systems for $\dps{1}$, $\dps{2}$, and $\dps{3}$, which are the blow-ups of the projective plane $\mbb{P}^2$ at respectively 1, 2, and 3 points in general position.}
\label{fig:dp_toricdiagsandweights}
\end{figure}
We also note the following relations between the toric divisors and the basis $(\hcl,\ecl_1,\ldots ,\ecl_n)$.
\be\nonumber
\begin{aligned}
\dps{1}: & \quad [\tc_1] = \ecl_1 \,, && [\tc_2] = \hcl-\ecl_1 \,, && && [\tc_4] = \hcl \,, && && [\tc_6] = \hcl-\ecl_1 \,, \\
\dps{2}: & \quad [\tc_1] = \ecl_1 \,, && [\tc_2] = \hcl-\ecl_1-\ecl_2 \,, && [\tc_3] = \ecl_2 \,, && [\tc_4] = \hcl-\ecl_2 \,, && && [\tc_6] = \hcl-\ecl_1 \,, \\
\dps{3}: & \quad [\tc_1] = \ecl_1 \,, && [\tc_2] = \hcl-\ecl_1-\ecl_2 \,, && [\tc_3] = \ecl_2 \,, && [\tc_4] = \hcl-\ecl_2-\ecl_3 \,, && [\tc_5] = \ecl_3 \,, && [\tc_6] = \hcl-\ecl_1-\ecl_3 \,. \\
\end{aligned}
\ee

\smlhdgnogap{Complete intersection representations}\\
\noindent
Del Pezzo surfaces can also be represented as complete intersections in products of projective spaces. Below we list a set of possible configuration matrices, where the entries in the first column specify the projective spaces whose product forms the ambient space, and each other column specifies the multi-degree of one of the defining equations.
\be
\begin{array}{llllllll}
\dps{1}\in&
\left[
\begin{array}{c | c}
\mathbb{P}^2 & 1 \\
\mathbb{P}^1 & 1
\end{array}
\right] &
\dps{2}\in &
\left[
\begin{array}{c | c c}
\mathbb{P}^2 & 1 & 1 \\
\mathbb{P}^1 & 1 & 0 \\
\mathbb{P}^1 & 0 & 1
\end{array}
\right] &
\dps{3}\in &
\left[
\begin{array}{c | c}
\mathbb{P}^1 & 1 \\
\mathbb{P}^1 & 1 \\
\mathbb{P}^1 & 1 
\end{array}
\right] &
\dps{4}\in &
\left[
\begin{array}{c | c}
\mathbb{P}^2 & 2 \\
\mathbb{P}^1 & 1
\end{array}
\right] \vspace{1mm}\\
\dps{5}\in&
\left[
\begin{array}{c | c c}
\mathbb{P}^4 & 2 & 2
\end{array}
\right]&
\dps{6}\in &
\left[
\begin{array}{c | c}
\mathbb{P}^3 & 3
\end{array}
\right] &
\dps{7}\in&
\left[
\begin{array}{c | c}
\mathbb{P}^2 & 2 \\
\mathbb{P}^1 & 2
\end{array}
\right]&
\dps{8}\in &
\left[
\begin{array}{c | c}
\mathbb{P}^2 & 3 \\
\mathbb{P}^1 & 1
\end{array}
\right] \,.
\end{array}
\ee
These representations are simple but they are not all favourable, that is, not every divisor class descends from one on the ambient space. This is easy to see since for example in the $\dps{8}$ case which has Picard number nine but only two divisors classes descend from the ambient space. It is possible to write down a favourable complete intersection representation in each case by sacrificing simplicity. In fact, the complete intersection\footnote{This representation reflects the fact that a blow-up of $\mbb{P}^2$ can be implemented by introducing a $\mbb{P}^1$ into the ambient space and providing a linear equation that allows `movement along' the $\mbb{P}^1$ only over the blow-up point.}
\be
\dps{n}\in
\left[
\begin{array}{ c | c c c c c }
\mathbb{P}^1 & 1 & 0 &  \cdots & 0 & 0 \\
\mathbb{P}^1 & 0 & 1 &  \cdots & 0 & 0 \\
\vdots & \vdots & \vdots & \ddots & \vdots & \vdots  \\
\mathbb{P}^1 & 0 & 0 &  \cdots & 1 & 0 \\
\mathbb{P}^1 & 0 & 0 &  \cdots & 0 & 1 \\
\mathbb{P}^2 & 1 & 1 &  \cdots & 1 & 1 \\
\end{array}
\right] \,,
\ee
where the top right block is an $n \times n$ identity matrix, provides a favourable realisation of all del Pezzo surfaces $\dps{n}$.

\smlhdg{Mori and nef cone generators}\\
\noindent
Finally, we list the Mori and nef cone generators for $\dps{n}$ which are required in the main text. 

The generators of the Mori cone of $\dps{n}$ are the exceptional curves\footnote{The two spaces $\dps{0}$ and $\dps{1}$ are exceptions: in these cases one has to include the hyperplane class $\hcl$ which is not exceptional as a generator of the Mori cone. From $\dps{2}$ onwards $\hcl$ is not needed as a generator since $\hcl-\ecl_i-\ecl_j$ are exceptional curves.}, and the generators of the dual nef cone can be found by standard algorithms for dual cones. For compactness of notation we will write the generators as vectors with respect to the standard basis $(l,\ecl_1,\ldots ,\ecl_n)$, that is, we write a divisor $D=k_0\hcl+k_1\ecl_1+\cdots +k_n\ecl_n$ as a vector $(k_0,k_1,\ldots ,k_n)$. Since the blow-ups are at general position, the classes $\ecl_i$ are on equal footing, and so the generator lists are invariant under their exchange. We, therefore, only list the generators subject to the ordering $k_1\geq k_2\geq\cdots\geq k_n$, keeping in mind that the other generators are obtained by permuting $k_1,\ldots ,k_n$. We also note that it is sufficient to list the generators for $\dps{8}$. The generators for $\dps{n}$ with $n<8$ can be obtained from the $\dps{8}$ ones by extracting the vectors with zero entries in the last $8-n$ positions and then removing those entries.

It turns out that the number of ordered generators for the Mori and nef cones for $\dps{8}$ is 7 and 50 respectively and, after including the permutations, these numbers increase to 240 and 19440, respectively. The list of $\dps{8}$ Mori cone generators is
\be\label{eq:dp_morigenlist}
\begin{array}{llll}
-\hat{\moricn}(\dps{8})= \{&\hspace{-3mm}(0,-1,0,0,0,0,0,0,0), & (-1,1,1,0,0,0,0,0,0), & (-2,1,1,1,1,1,0,0,0), \\
&\hspace{-3mm}(-3,2,1,1,1,1,1,1,0) & (-4,2,2,2,1,1,1,1,1), & (-5,2,2,2,2,2,2,1,1), \\
&\hspace{-3mm}(-6,3,2,2,2,2,2,2,2),&\ldots\quad\}\, .
\end{array}
\ee
The $\dps{8}$ nef cone generators are
\be
\begin{aligned}
-\hat{\nefcn}(\dps{8})= \{~
(-1,1,0,0,0,0,0,0,0), & &
(-1,0,0,0,0,0,0,0,0), \phantom{\}}& &
(-2,1,1,1,1,0,0,0,0), \\
(-2,1,1,1,0,0,0,0,0), & &
(-3,2,1,1,1,1,1,0,0), \phantom{\}}& &
(-3,2,1,1,1,1,0,0,0), \\
(-4,3,1,1,1,1,1,1,1), & &
(-4,3,1,1,1,1,1,1,0), \phantom{\}}& &
(-4,2,2,2,1,1,1,1,0), \\
(-4,2,2,2,1,1,1,0,0), & &
(-5,3,2,2,2,1,1,1,1), \phantom{\}}& &
(-5,3,2,2,2,1,1,1,0), \\
(-5,2,2,2,2,2,2,1,0), & &
(-5,2,2,2,2,2,2,0,0), \phantom{\}}& &
(-6,4,2,2,2,2,1,1,1), \\
(-6,3,3,3,2,1,1,1,1), & &
(-6,3,3,2,2,2,2,1,1), \phantom{\}}& &
(-6,3,3,2,2,2,2,1,0), \\
(-7,4,3,3,2,2,2,1,1), & &
(-7,4,3,2,2,2,2,2,2), \phantom{\}}& &
(-7,3,3,3,3,2,2,2,1), \\
(-7,3,3,3,3,2,2,2,0), & &
(-8,5,3,3,2,2,2,2,2), \phantom{\}}& &
(-8,4,4,3,3,2,2,2,1), \\
(-8,4,3,3,3,3,3,1,1), & &
(-8,4,3,3,3,3,2,2,2), \phantom{\}}& &
(-8,3,3,3,3,3,3,3,1), \\
(-8,3,3,3,3,3,3,3,0), & &
(-9,5,4,3,3,3,2,2,2), \phantom{\}}& &
(-9,4,4,4,4,2,2,2,2), \\
(-9,4,4,4,3,3,3,2,1), & &
(-9,4,4,3,3,3,3,3,2), \phantom{\}}& &
(-10,6,3,3,3,3,3,3,3), \\
(-10,5,5,3,3,3,3,3,2), & &
(-10,5,4,4,4,3,3,2,2), \phantom{\}}& &
(-10,4,4,4,4,4,3,3,1), \\
(-10,4,4,4,4,3,3,3,3), & &
(-11,6,4,4,4,3,3,3,3), \phantom{\}}& &
(-11,5,5,4,4,4,3,3,2), \\
(-11,4,4,4,4,4,4,4,3), & &
(-12,6,5,4,4,4,4,3,3), \phantom{\}}& &
(-12,5,5,5,5,4,3,3,3), \\
(-12,5,5,5,4,4,4,4,2), & &
(-13,6,6,4,4,4,4,4,4), \phantom{\}}& &
(-13,6,5,5,5,4,4,4,3), \\
(-14,6,6,5,5,5,4,4,4), & &
(-14,6,5,5,5,5,5,5,3), \phantom{\}}& &
(-15,6,6,6,5,5,5,5,4), \\
(-16,6,6,6,6,6,5,5,5), & &
(-17,6,6,6,6,6,6,6,6), \phantom{\}}& &\ldots\hspace{3.2cm}\}\, .\\
\end{aligned}
\label{eq:dp_nefgenlist}
\ee
For either list, the dots indicate the vectors obtained from the ones listed by permuting the last eight entries.


\section{Proof of the main theorem}\label{app:prfsketch}
In this appendix we sketch an alternative proof for Theorem~\ref{thm:shift}. A more rigorous proof, using a somewhat different method based on linear systems of divisors, can be found in Ref.~\cite{mathpaper}.\\[2mm]
We start with a smooth compact complex projective surface $\surf$ and a divisor $D\subset S$ with associated line bundle $\mc{O}_\surf(D)$. Further, we have an irreducible curve $C\subset S$ which intersects $D$ negatively, so $D\cdot C <0$. (This implies that $C^2<0$.) Then we have the short exact Koszul sequence
\begin{equation}
 0\rightarrow \mc{O}_\surf(D-C)\rightarrow \mc{O}_\surf(D)\rightarrow \mc{O}_\surf(D)|_C \rightarrow 0 \label{KosCD}
\end{equation} 
for the restriction $\mc{O}_\surf(D)|_C$ of the line bundle $\mc{O}_\surf(D)$ to the curve $C$. The degree of this restricted line bundle is negative since ${\rm deg}\, \mc{O}_\surf(D)|_C=D\cdot C<0$ . It is well-known that negative degree line bundles on curves have a vanishing zeroth cohomology (see, for example, Ref.~\cite{griffiths2014principles}), that is, $H^0(C,\mc{O}_\surf(D)|_C)=0$. This vanishing, together with the long exact sequence associated to the Koszul sequence~\eqref{KosCD}, immediately implies that
\begin{equation}
 h^0(\surf,\mc{O}_\surf(D-C))= h^0(\surf,\mc{O}_\surf(D))\, .
\end{equation} 
This means taking away from the divisor $D$ an irreducible curve $C$ which intersects $D$ negatively does not change the dimension of the zeroth cohomology. Hence, we have
\begin{lemma} 
Let $\surf$ be a smooth compact complex projective surface, $D\subset S$ an effective divisor and $C\subset S$ an irreducible curve with $D\cdot C<0$. Then we have
\begin{equation}
 h^0(\surf,\mc{O}_\surf(D-C))= h^0(\surf,\mc{O}_\surf(D))\, .
\end{equation} 
\end{lemma}
\noindent This completes the first step of the argument.\\[2mm]
Next, we would like to iterate this process and work out how many times $C$ can be subtracted from $D$ without changing the zeroth cohomology dimension. Define the sequence of divisors $D_k=D-kC$, for $k\geq0$. From the above argument, the two divisors $D_k$ and $D_{k+1}$ have the same zeroth cohomology dimension if $D_k\cdot C<0$. This condition immediately translates into
\begin{equation}
 k<\frac{D\cdot C}{C^2}
\end{equation}
This means we have the following 
\begin{lemma}\label{lemmaDC}
 Let $\surf$ be a smooth compact complex projective surface, $D\subset S$ an effective divisor and $C\subset S$ an irreducible curve with $D\cdot C<0$. Define the divisors $D_k:=D-k\, C$, where $k=0,1,\ldots ,n$ and
\begin{equation}
 n={\rm ceil}\left(\frac{D\cdot C}{C^2}\right)\, .
\end{equation}
Then, $h^0(\surf,\mc{O}_\surf(D_k))=h^0(\surf,\mc{O}_\surf(D))$ for all $k=0,1,\ldots n$. (Here, ${\rm ceil}$ is the ceiling function.)
\end{lemma} 
\noindent In particular, this lemma implies that the divisor
\begin{equation}
 D_n=D-{\rm ceil}\left(\frac{D\cdot C}{C^2}\right)C\, ,
\end{equation} 
obtained by taking away the largest possible multiple of $C$, leads to the same zeroth cohomology dimension as $D$.\\[2mm]
The next step is to apply Lemma~\ref{lemmaDC} repeatedly, starting with the effective divisor $D\subset S$ and a number of irreducible curves $C_i$ with $D\cdot C_i<0$, where $i=1,\ldots ,N$. This leads to
\begin{lemma}\label{lemmaCi}
Let $\surf$ be a smooth compact complex projective surface, $D\subset S$ an effective divisor and $C_i\subset S$ irreducible (pairwise different) curves with $D\cdot C_i<0$, where $i=1,\ldots ,N$. Define the divisors $D_{(i)}$, where $i=0,1,\ldots ,N$, recursively by $D_{(0)}=D$ and $D_{(i)}=D_{(i-1)}-k_iC_i$, where $k_i\in\{0,1,\ldots ,m_i\}$ but otherwise arbitrary and
\begin{equation}
 m_i={\rm ceil}\left(\frac{D_{(i-1)}\cdot C_i}{C_i^2}\right)\, . \label{midef}
\end{equation}
Then $h^0(\surf,\mc{O}_\surf(D_{(i)}))=h^0(\surf,\mc{O}_\surf(D))$ for all $i=0,1,\ldots ,N$. 
\end{lemma}
\begin{proof} We proceed by induction on $i$ and first note that the statement is clearly true for $i=0$, since $D_{(0)}=D$. Suppose the statement is true for all $j=0,1,\ldots ,i-1$. From $D_{(i-1)}=D-\sum_{j=1}^{i-1}k_jC_j$ it follows that
\begin{equation}
 D_{(i-1)}\cdot C_i=D\cdot C_i-\sum_{j=1}^{i-1}k_j\,C_j\cdot C_i\, .
\end{equation} 
The first term on the right-hand side is negative by assumption and the second term is less equal than zero, since $C_j\cdot C_i\geq 0$ and $k_j\geq 0$. It follows that $D_{(i-1)}\cdot C_i<0$ and we can, hence, apply Lemma~\ref{lemmaDC} with $D=D_{(i-1)}$ and $C=C_i$. This implies that $h^0(\surf,\mc{O}_\surf(D_{(i)}))=h^0(\surf,\mc{O}_\surf(D_{(i-1)}))=h^0(\surf,\mc{O}_\surf(D))$.
\end{proof}
\noindent In particular, the sequence $D_{(i)}$ of divisors defined by $D_{(0)}=D$ and $D_{(i)}=D_{(i-1)}-m_iC_i$, with $m_i$  given in Eq.~\eqref{midef}, where the largest possible multiple of the $C_i$ is subtracted at each step, satisfies $h^0(\surf,\mc{O}_\surf(D_{(i)}))=h^0(\surf,\mc{O}_\surf(D))$ for all $i=0,1,\ldots ,N$. \\[2mm]
For our application in the main text, we require a statement slightly weaker than Lemma~\ref{lemmaCi} which is not recursive in nature. To this end, we consider
\begin{lemma} 
Let $\surf$ be a smooth compact complex projective surface, $D\subset S$ an effective divisor and $C_i\subset S$ irreducible (pairwise different) curves with $D\cdot C_i<0$, where $i=1,\ldots ,N$. Define the divisors $D_{(i)}$, where $i=0,1,\ldots ,N$, recursively by $D_{(0)}=D$ and $D_{(i)}=D_{(i-1)}-n_iC_i$, where
\begin{equation}
 n_i={\rm ceil}\left(\frac{D\cdot C_i}{C_i^2}\right)\, . \label{nidef}
\end{equation}
Then $h^0(\surf,\mc{O}_\surf(D_{(i)}))=h^0(\surf,\mc{O}_\surf(D))$ for all $i=0,1,\ldots ,N$. 
\end{lemma}
\begin{proof}
The recursive definition implies that $D_{(i-1)}=D-\sum_{j=1}^{i-1}n_iC_i$, so that
\begin{equation}
 D_{(i-1)}\cdot C_i=D\cdot C_i-\sum_{j=1}^{i-1}n_j\, C_j\cdot C_i\, .
\end{equation}
This equation shows that $n_i\leq m_i$ with  $n_i$ and $m_i$ defined in Eqs.~\eqref{nidef} and \eqref{midef}, respectively. Hence, Lemma~\ref{lemmaCi} implies the desired result.
\end{proof}
\noindent In particular, the last lemma shows that the divisor
\begin{equation}
 \tilde{D}:=D-\sum_{C_i\,{\rm irr.}}\theta(-D\cdot C_i)\,{\rm ceil}\left(\frac{D\cdot C_i}{C_i^2}\right)C_i
\end{equation}
 satisfies $h^0(\surf,\mc{O}_\surf(\tilde{D}))=h^0(\surf,\mc{O}_\surf(D))$. (Here, $\theta$ is the Heaviside function which has been included in the sum in order to ensure that only irreducible curves $C_i$ with $D\cdot C_i<0$ contribute.) This is precisely the statement we wanted to prove.

\newpage

\bibliographystyle{JHEP}

\providecommand{\href}[2]{#2}\begingroup\raggedright\endgroup


\begin{thebibliography}{10}

\bibitem{Constantin:2018hvl}
A.~Constantin and A.~Lukas, \emph{{Formulae for Line Bundle Cohomology on
  Calabi-Yau Threefolds}},  \href{https://arxiv.org/abs/1808.09992}{{\ttfamily
  1808.09992}}.

\bibitem{Constantin:2018otr}
A.~Constantin, \emph{{Heterotic String Models on Smooth Calabi-Yau
  Threefolds}}.
\newblock PhD thesis, Oxford U., 2018.
\newblock \href{https://arxiv.org/abs/1808.09993}{{\ttfamily 1808.09993}}.

\bibitem{Buchbinder:2013dna}
E.~I. Buchbinder, A.~Constantin and A.~Lukas, \emph{{The Moduli Space of
  Heterotic Line Bundle Models: a Case Study for the Tetra-Quadric}},
  \href{http://dx.doi.org/10.1007/JHEP03(2014)025}{\emph{JHEP} {\bfseries 1403}
  (2014) 025}, [\href{https://arxiv.org/abs/1311.1941}{{\ttfamily 1311.1941}}].

\bibitem{Klaewer:2018sfl}
D.~Klaewer and L.~Schlechter, \emph{{Machine Learning Line Bundle Cohomologies
  of Hypersurfaces in Toric Varieties}},
  \href{http://dx.doi.org/10.1016/j.physletb.2019.01.002}{\emph{Phys. Lett.}
  {\bfseries B789} (2019) 438--443},
  [\href{https://arxiv.org/abs/1809.02547}{{\ttfamily 1809.02547}}].

\bibitem{Larfors:2019sie}
M.~Larfors and R.~Schneider, \emph{{Line bundle cohomologies on CICYs with
  Picard number two}},  \href{https://arxiv.org/abs/1906.00392}{{\ttfamily
  1906.00392}}.

\bibitem{ml}
C.~R. Brodie, A.~Constantin, R.~Deen and A.~Lukas, \emph{{Machine Learning Line
  Bundle Cohomology}},  \href{https://arxiv.org/abs/in preparation}{{\ttfamily
  in preparation}}.

\bibitem{mathpaper}
C.~R. Brodie, A.~Constantin, R.~Deen and A.~Lukas, \emph{{Topological Formulae
  for Line Bundle Cohomology on Surfaces}},  \href{https://arxiv.org/abs/in
  preparation}{{\ttfamily in preparation}}.

\bibitem{lazarsfeld2004positivity}
R.~Lazarsfeld, \emph{Positivity in algebraic geometry}.
\newblock Springer, Berlin, 2004.

\bibitem{Anderson:2007nc}
L.~B. Anderson, Y.-H. He and A.~Lukas, \emph{{Heterotic Compactification, An
  Algorithmic Approach}},
  \href{http://dx.doi.org/10.1088/1126-6708/2007/07/049}{\emph{JHEP} {\bfseries
  07} (2007) 049}, [\href{https://arxiv.org/abs/hep-th/0702210}{{\ttfamily
  hep-th/0702210}}].

\bibitem{Gray:2007yq}
J.~Gray, Y.-H. He, A.~Ilderton and A.~Lukas, \emph{{A New Method for Finding
  Vacua in String Phenomenology}},
  \href{http://dx.doi.org/10.1088/1126-6708/2007/07/023}{\emph{JHEP} {\bfseries
  07} (2007) 023}, [\href{https://arxiv.org/abs/hep-th/0703249}{{\ttfamily
  hep-th/0703249}}].

\bibitem{Anderson:2008uw}
L.~B. Anderson, Y.-H. He and A.~Lukas, \emph{{Monad Bundles in Heterotic String
  Compactifications}},
  \href{http://dx.doi.org/10.1088/1126-6708/2008/07/104}{\emph{JHEP} {\bfseries
  07} (2008) 104}, [\href{https://arxiv.org/abs/0805.2875}{{\ttfamily
  0805.2875}}].

\bibitem{He:2009wi}
Y.-H. He, S.-J. Lee and A.~Lukas, \emph{{Heterotic Models from Vector Bundles
  on Toric Calabi-Yau Manifolds}},
  \href{http://dx.doi.org/10.1007/JHEP05(2010)071}{\emph{JHEP} {\bfseries 05}
  (2010) 071}, [\href{https://arxiv.org/abs/0911.0865}{{\ttfamily 0911.0865}}].

\bibitem{Anderson:2009mh}
L.~B. Anderson, J.~Gray, Y.-H. He and A.~Lukas, \emph{{Exploring Positive Monad
  Bundles And A New Heterotic Standard Model}},
  \href{http://dx.doi.org/10.1007/JHEP02(2010)054}{\emph{JHEP} {\bfseries 02}
  (2010) 054}, [\href{https://arxiv.org/abs/0911.1569}{{\ttfamily 0911.1569}}].

\bibitem{CohomOfLineBundles:Algorithm}
R.~Blumenhagen, B.~Jurke, T.~Rahn and H.~Roschy, \emph{{Cohomology of Line
  Bundles: A Computational Algorithm}},
  \href{http://dx.doi.org/10.1063/1.3501132}{\emph{J.~Math.~Phys.} {\bfseries
  51} (2010) 103525}, [\href{https://arxiv.org/abs/1003.5217}{{\ttfamily
  1003.5217}}].

\bibitem{CohomOfLineBundles:Proof}
T.~Rahn and H.~Roschy, \emph{{Cohomology of Line Bundles: Proof of the
  Algorithm}}, \href{http://dx.doi.org/10.1063/1.3501135}{\emph{J.~Math.~Phys.}
  {\bfseries 51} (2010) 103520},
  [\href{https://arxiv.org/abs/1006.2392}{{\ttfamily 1006.2392}}].

\bibitem{Jow:2011}
S.-Y. {Jow}, \emph{{Cohomology of toric line bundles via simplicial Alexander
  duality}}, \href{http://dx.doi.org/10.1063/1.3562523}{\emph{Journal of
  Mathematical Physics} {\bfseries 52} (Mar, 2011) 033506--033506},
  [\href{https://arxiv.org/abs/1006.0780}{{\ttfamily 1006.0780}}].

\bibitem{Blumenhagen:2008zz}
R.~Blumenhagen, V.~Braun, T.~W. Grimm and T.~Weigand, \emph{{GUTs in Type IIB
  Orientifold Compactifications}},
  \href{http://dx.doi.org/10.1016/j.nuclphysb.2009.02.011}{\emph{Nucl. Phys.}
  {\bfseries B815} (2009) 1--94},
  [\href{https://arxiv.org/abs/0811.2936}{{\ttfamily 0811.2936}}].

\bibitem{Hartshorne1977}
R.~Hartshorne, \emph{Algebraic geometry}.
\newblock Springer Science+Business Media, Inc, New York, 2010.

\bibitem{griffiths2014principles}
P.~Griffiths and J.~Harris, \emph{Principles of Algebraic Geometry}.
\newblock Wiley Classics Library. Wiley, 2014.

\bibitem{cox2011toric}
D.~Cox, J.~Little and H.~Schenck, \emph{Toric Varieties}.
\newblock Graduate Studies in Mathematics. American Mathematical Soc., 2011.

\bibitem{Anderson:2011ty}
L.~B. Anderson, J.~Gray, A.~Lukas and B.~Ovrut, \emph{{The Atiyah Class and
  Complex Structure Stabilization in Heterotic Calabi-Yau Compactifications}},
  \href{http://dx.doi.org/10.1007/JHEP10(2011)032}{\emph{JHEP} {\bfseries 10}
  (2011) 032}, [\href{https://arxiv.org/abs/1107.5076}{{\ttfamily 1107.5076}}].

\bibitem{Demazure1976}
M.~Demazure, \emph{Surfaces de {D}el {P}ezzo: {II} - Éclater $n$ points dans
  $\mathbb{P}^2$}, {\emph{{S}éminaire sur les singularités des surfaces}
  (1976-1977) 1--13}.

\end{thebibliography}

\end{document}